\documentclass[11pt]{article}
\usepackage{graphicx}
\usepackage{subfigure}
\usepackage{color}
\usepackage{algorithm,algorithmic}
\usepackage{amsthm,amsmath,amsfonts}
\usepackage{array}
\usepackage{verbatim}
\usepackage{graphics}
\usepackage{epstopdf}
\usepackage{multicol}
\usepackage{url}
\usepackage{fullpage}

\newtheorem{theorem}{Theorem}[section]

\newtheorem{corollary}[theorem]{Corollary}

\begin{document}
\title{Backward Path Growth for Efficient Mobile Sequential Recommendation}
\author{Jianbin Huang, Xuejun Huangfu, Heli Sun, Hong Cheng and Qinbao Song
\thanks{J. Huang and X. Huangfu are with the School of Software, Xidian University, Xi'an, China, 710071. H. Sun and Q. Song are with the Department of Computer Science and Technology, Xi'an Jiaotong University, China, 710049. H. Cheng is with the Chinese University of Hong Kong, Shatin, N.T., Hong Kong, China. Corresponding author: Jianbin Huang (\texttt {jbhuang@xidian.edu.cn})}}
\maketitle

\begin{abstract}
\noindent
The problem of mobile sequential recommendation is presented to suggest a route connecting some pick-up points for a taxi driver so that he/she is more likely to get passengers with less travel cost. Essentially, a key challenge of this problem is its high computational complexity. In this paper, we propose a dynamical programming based method to solve this problem. Our method consists of two separate stages: an offline pre-processing stage and an online search stage. The offline stage pre-computes optimal potential sequence candidates from a set of pick-up points, and the online stage selects the optimal driving route based on the pre-computed sequences with the current position of an empty taxi. Specifically, for the offline pre-computation, a backward incremental sequence generation algorithm is proposed based on the iterative property of the cost function. Simultaneously, an incremental pruning policy is adopted in the process of sequence generation to reduce the search space of the potential sequences effectively. In addition, a batch pruning algorithm can also be applied to the generated potential sequences to remove the non-optimal ones of a certain length. Since the pruning effect continuously increases with the increase of the sequence length, our method can search the optimal driving route efficiently in the remaining potential sequence candidates. Experimental results on real and synthetic data sets show that the pruning percentage of our method is significantly improved compared to the state-of-the-art methods, which makes our method can be used to handle the problem of mobile sequential recommendation with more pick-up points and to search the optimal driving routes in arbitrary length ranges.

\noindent
\emph{Key words:}	Mobile Sequential Recommendation, Potential Travel Distance, Backward Path Growth, Sequence Pruning.
\end{abstract}

\section{Introduction}

With the wide utilization of the sensor, wireless communication and information infrastructures such as GPRS, WiFi and RFID, we can easily access the location trace data for a large number of moving objects. Finding useful knowledge from these trajectory data will provide strong support for the real-time decision and the intelligence services in the related applications \cite{1}. Reducing taxicab cruising cost problem is a typical example \cite{2,3}. An unloaded taxi driving on the road not only leads to waste of fuel and time, but also may result in traffic congestion. However, some high probability pick-up points in the taxi trajectory data of the high-yield drivers can be excavated to guide new drivers to pick up passengers in a more economical and efficient way. Therefore, high-efficiency mobile pattern mining and recommendation algorithm can improve business performance of the drivers and reduce the energy consumption. This is a problem possessing considerable theoretical significance and applicable values \cite{3,4}.

In \cite{2}, Ge \emph{et al.} have proposed a novel problem of Mobile Sequential Recommendation (MSR), which is to suggest a route connecting some pick-up points for an empty cab so that the driver is more likely to get passengers with less travel cost starting from its current position. It is a challenging task, because we need to enumerate and compare all possible routes derived from the given set of pick-up points which involves a rather high computational complexity. To solve the MSR problem, they provided a function of Potential Travel Distance (PTD) for evaluating the cost of a driving route. Essentially, the PTD value of a suggested route is the expected travel distance for an empty cab before it successfully gets new passengers when it travels along the route. To reduce the computational cost, two effective potential sequence pruning algorithms LCP and SkyRoute, which are based on the monotone property of the PTD function, have been proposed in \cite{2}. However, the time and space complexities of these two algorithms both grow exponentially with the number of pick-up points and the length of the suggested driving route, so they can only perform the driving route recommendation with a length constraint in a small number of pick-up points.

However, in real applications, a driver always wants to obtain the optimal driving routes in a range of length, so that he/she can select a preferable driving route among them. In this paper, we consider a generalized mobile sequential recommendation problem with minimal and maximal length constraints. We propose a solution including an offline stage and an online stage. The offline stage effectively prunes the search space and generates a small set of sequence candidates. The online stage is for obtaining the optimal driving route given the current position of an unloaded taxi as the starting point. Specifically, for the offline pre-computation, we have deeply studied the nature of the PTD function and have found that it satisfies the iterative calculation feature. This feature allows us to incrementally construct a potential driving route backward from the terminal point to the starting point. Based on the above calculation feature of the PTD function, we have also found that a set of potential sequences with the same length and the same starting point satisfies the incremental and batch pruning properties. Then, we design a novel mobile sequential recommendation method which takes full advantage of the iterative nature of the PTD function. It incrementally generates potential sequences and removes a lot of impossible search space in the process which greatly enhances the time efficiency and reduces the memory consumption. Among the generated potential sequences with the same length, we can still remove a large number of potential sequences which cannot form the optimal route by using a batch pruning policy. It can dramatically reduce the number of the remaining sequence candidates. Experimental results show that the offline pruning effect and the online search efficiency of our method are significantly improved compared to the existing state-of-the-art methods.

The main contributions of the paper are given as follows: 1) Our algorithm can generate all possible sequence candidates of arbitrary length which can be used to suggest the driving route with any length range constraint; 2) The recursive formula of the PTD function  is presented which makes the incremental generation of the potential sequences possible; 3) A backward incremental sequence generation algorithm with less time and a smaller space complexity is proposed; 4) An efficient method for comparing the PTD cost of different potential sequences and driving routes is presented; 5) An effective sequence pruning method combining incremental pruning and batch pruning is adopted which significantly improves the offline pruning effect.

The rest of the paper is organized as follows. Section 2 introduces the background and the related work. Section 3 gives the iterative nature of the PTD function and the proposed sequence pruning principle. In Section 4, the offline sequence generation and online search algorithms are described in detail. Section 5 gives the experimental results and analysis. Section 6 discusses some extension of our method. Finally, section 7 concludes the paper.

\section{BACKGROUND} \label{sec:formal}

In this section, we first introduce the MSR problem and then describe the previous works.

\subsection{Problem Statement}

Let $c_i$ be a potential pick-up position and \(C = \{ {c_1},{c_2},...,{c_N}\} \) be a set of $N$ pick-up points. The probability that a taxi can successfully carry passengers at the pick-up point \(c_i\) is denoted by \(P(c_i)\), and the set of mutually independence probability is \(P = \{ P({c_1}),P({c_2}),...,P({c_N})\}\). Which driving route will lead to the minimum cost of picking up a new passenger when a taxi travels all or part of the pick-up points in $C$ starting from its current location? This is the MSR problem introduced by Ge \emph{et al.} \cite{2}. The problem can also be found in other scenarios such as recommending tourist routes, searching parking places, \emph{etc}. In the following, we introduce some concepts of the MSR problem and all the symbols used in this paper are described in Table~\ref{table1}.
\begin{table*}
\footnotesize
\caption{Adopted symbols.\label{table1}}
\begin{center}
\begin{tabular}{|c|p{13cm}|} \hline
Symbols & Definition \\ \hline
$C$ & A set of potential pick-up points. \\ \hline
\({c_i}\)& A location point. It represents the current location of the cab when \(i = 0\) and a pick-up point in $C$ with \(i > 0\).\\ \hline
$N$ &The number of pick-up points in $C$. \\ \hline
\(P({c_i})\) &The probability of successfully taking passenger at \({c_i}\).\\ \hline
$D$ & The distance matrix of pairs of location points. \\ \hline
\({D_{{c_i},{c_j}}}\)& The distance from \({c_i}\) to \({c_j}\). \\ \hline
\(\vec r\)& The potential mobile sequence containing one or more different pick-up points. \\ \hline
\(\left\| {\vec r} \right\|\) & The length of potential sequence \(\vec r\) (i.e., the number of different pick-up points in \(\vec r\)).\\ \hline
\({C_{\vec r}}\) & The set of pick-up points in the potential sequence \(\vec r\).\\ \hline
\({P_{\vec r}}\) & The probability vector of the pick-up points consisting of the potential sequence \(\vec r\).\\ \hline
\(\left| A \right|\) & The number of elements in the set $A$. \\ \hline
\( \left\langle c,\vec r \right\rangle \) & A driving route that travels the potential sequence \(\vec r\) starting from the location point $c$. \\ \hline
\(s(\vec r)\) & The source point of the potential sequence \(\vec r\).\\ \hline
\(\overrightarrow R \) & A set of all potential sequences.\\ \hline
\(\overrightarrow {{R^L}} \) & A set of potential sequences with length $L$.\\ \hline
\(\overrightarrow {R_{\vec r}^L} \) & A set of potential sequences that have the same length, source point and pick-up point set as \(\vec r\).\\ \hline
\(\overrightarrow {R_c^L} \) & A set of potential sequences with length $L$ and source point $c$.\\ \hline
\end{tabular}
\end{center}
\end{table*}

 Let \(\vec r = \left\langle {{c_1},{c_2}, \cdots ,{c_L}} \right\rangle \) be a potential sequence with length $L$ derived from the pick-up points set \(C\), where each \(c_i\) in \(\vec r\) is different from each other. ${c_1}$ is called the source point of $\vec r$ and ${c_L}$ is called the destination point. \({C_{\vec r}} = \{ {c_1},{c_2}, \cdots ,{c_L}\} \) denotes the pick-up points set of the potential sequence \(\vec r\). \(\overrightarrow R  = \{ \vec r|{C_{\vec r}} \subseteq C \wedge {C_{\vec r}} \ne \emptyset \} \) is the set of all potential sequences derived from \(C\). \(\left| {\overrightarrow R } \right| = M\) is the number of all possible potential sequences in $\overrightarrow R$. \(P(\vec r) = \left\langle {P({c_1}),P({c_2}), \cdots ,P({c_L})} \right\rangle \) is the probability vector of the potential sequence \(\vec r\) consisting of the probabilities of all the pick-up points in \(\vec r\). \(\vec d = \left\langle {{c_0},\vec r} \right\rangle \) is a driving route, where \({c_0}\) is the current location of a taxi, \(\vec r\) is the sequence of pick-up points, and \(\left\| {\vec d} \right\| = \left\| {\vec r} \right\| = k\) is the length of \(\vec d\).

For a driving route \(\vec d = \left\langle {{c_0},\vec r} \right\rangle \), a PTD function is defined in \cite{2} to evaluate its travel cost. Let $
\small
D(\vec d) = \left\langle \begin{array}{l}
 D_{c_0 ,c_1 } ,(D_{c_0 ,c_1 }  + D_{c_1 ,c_2 } ), \ldots ,\sum\limits_{i=1}^L D_{c_{i-1} ,c_{i} },D_\infty  \\
 \end{array} \right\rangle$ be the distance vector of $\vec d$ and probability vector $
\small
P(\vec d) = \left\langle \begin{array}{l}
 P(c_1 ),\overline {P(c_1 )}  \cdot P(c_2 ), \ldots ,\prod\limits_{i=1}^{L-1} \overline {P({c_i})}
 \cdot P(c_L ), \prod\limits_{i=1}^L \overline {P({c_i})}  \\
 \end{array} \right\rangle$, then the PTD cost of \(\vec d\)
can be calculated by
\begin{equation}
\label{equ1}
{F(\vec d) = F({c_0},\vec r,P(\vec r)) = D(\vec d) \cdot P(\vec d)},
\end{equation}
where \({D_\infty}\) represents the desired maximum cruising distance of a driver for picking up new passengers, and it can be manually specified. The PTD value of a driving route \(\vec d\) represents the expected travel distance of an empty cab for picking up new passengers when it is driving along the route. The smaller the PTD cost of a driving route, the shorter travel distance and the less required energy and cost for the cab to take new guests driving along it.

The objective of the simple MSR problem is to recommend a driving route derived from the set of pick-up points $C$ for a cab driver, so that the expected potential travel distance (PTD) for finding new passengers is minimal. An illustration example is shown in Figure~\ref{fig1}, there are two different driving routes $\vec d_1 =\left\langle {c_0, c_1, c_2} \right\rangle$ and $\vec d_2 =\left\langle c_0, c_2, c_3 \right\rangle$ with length 2. Let ${D_\infty} = 10$. We can get that \(D(\vec d_1) = \left\langle {{D_{{c_0},{c_1}}},({D_{{c_0},{c_1}}} + {D_{{c_1},{c_2}}}}), {D_\infty} \right\rangle = \left\langle 2, 7, 10 \right\rangle \), \( P(\vec d_1) = \left\langle {P({c_1})}, \overline {P({c_1})} \cdot {P({c_2})},  \overline {P({c_1})} \cdot \overline {P({c_2})} \right\rangle = \left\langle 0.5, 0.15, 0.35 \right\rangle \), \(D(\vec d_2) = \left\langle {{D_{{c_0},{c_2}}},({D_{{c_0},{c_2}}} + {D_{{c_2},{c_3}}}}), {D_\infty} \right\rangle = \left\langle 4, 5, 10 \right\rangle \), and \\
\( P(\vec d_2) = \left\langle P({c_2}), \overline {P({c_2})} \cdot {P({c_3})}, \overline {P({c_2})} \cdot \overline {P({c_3})} \right\rangle = \left\langle 0.3, 0.56, 0.14 \right\rangle \). So the PTD cost of $\vec d_1$ is $F(\vec d_1) = 2 \times 0.5 + 7 \times 0.15 + 10 \times 0.35 =5.55$ and the PTD cost of $\vec d_2$ is $F(\vec d_2) = 4 \times 0.3 + 5 \times 0.56 + 10 \times 0.14 =5.4$. We can see that the PTD cost of $\vec d_2$ is smaller than that of $\vec d_1$ and then $\vec d_2$ should be recommended.
\begin{figure}[!t]
\centering
\includegraphics[width=3.2in]{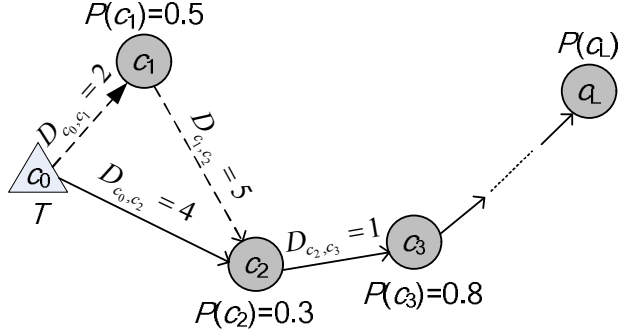}
\caption{An example of the driving route and its PTD cost.\label{fig1}}
\end{figure}

Since the computational complexity of the simple MSR problem is $O(N!)$ \cite{2}, a brute-force method for searching the optimal route in $\overrightarrow R$ is inefficient. In \cite{2}, Ge \emph{et al.} focus on the MSR problem with a length constraint due to the high complexity of the simple MSR problem.

However, in real life, a user usually prefers to request a route within a length range. For example, a cab driver wants to get an optimal driving route in the nearby area with length between 3 and 5. In this case, a recommendation method with a length constraint will be inefficient in handling such a service request while a recommendation method with unconstraint simple MSR problem is also inefficient when the suggested length is less than 3 or more than 5. Therefore, we focus on a more general MSR problem in this paper with length between the minimum $L_{min}$ and the maximum $L_{max}$. The generalized MSR problem is given as follows.

\begin{center}
\begin{tabular}[!thp]{|l|} \hline
\footnotesize\textbf{The generalized MSR problem}\\

\vspace{-0.5cm}
\\
\footnotesize\textbf{Given:} \\ \vspace{-0.1cm}
\footnotesize A set of potential pick-up points \(C = \{ {c_1},{c_2},...,{c_N}\} \);\\ \vspace{-0.1cm}
\footnotesize A probability set \(P = \{ P({c_1}),P({c_2}),...,P({c_N})\} \);\\ \vspace{-0.1cm}
\footnotesize A potential sequence set \(\overrightarrow R  = \{ {\vec r_1},{\vec r_2},...,{\vec r_M}\} \); \\ \vspace{-0.1cm}
\footnotesize The position \({c_0}\) of a cab who needs the service; \\ \vspace{-0.1cm}
\footnotesize　The suggested minimal length \({L_{\min }} \in \{ 1,2, \ldots ,N\} \); \\ \vspace{-0.1cm}
\footnotesize　The suggested maximal length \({L_{\max }} \in \{ {L_{\min }}, \ldots ,N\} \).\\ \vspace{-0.1cm}
\vspace{-0.3cm}
\\
\footnotesize\textbf{Objective:} Recommending an optimal driving route \\
\footnotesize　\(\vec d =  < {c_0},\vec r > \), s.t. \\
\footnotesize　\hspace{8em}\(\mathop {\min }\limits_{\vec r \in \overrightarrow R } \)\(F({c_0},\vec r,P(\vec r))\),\\
\footnotesize　where \(\vec r \in \overrightarrow R \) and \({L_{\min }} \le \left\| {\vec d} \right\| \le {L_{\max }}\).\\  \hline
\end{tabular}
\end{center}

Actually, the above problem is a computational extension of the simple MSR problem with more flexible parameter specification. When we set \({L_{\min }} = 1\) and \({L_{\max}} = N\), it is the simple MSR problem. Moreover, the length constrained MSR problem can be obtained by setting \({L_{\min }} = {L_{\max }} = L\) \cite{2}. In this paper, we present a method to handle any cases of \(1 \le {L_{\min }} \le {L_{\max }} \le N\). In particular, in order to compare the cost of potential sequences in various lengths, we set the \({D_\infty}\) to be equal for all the suggested routes of arbitrary length.

Since the number of  driving routes satisfying the length constraint from \(L_{\min}\) to \(L_{\max}\) is \(\sum\limits_{L = {L_{\min }}}^{{L_{\max }}} {\left( {\begin{array}{*{20}{c}}
N\\
L
\end{array}} \right)}  \cdot L!\), the computational complexity of the generalized MSR problem is no more than the complexity of the simple MSR problem $O(N!)$ and is no less than the complexity of the MSR problem with fixed length $L_{max}$. Therefore, it cannot be effectively solved by the brute-force search method.

\subsection{Related Work}

In recent years, intelligent transportation systems and trajectory data mining have aroused widespread attentions \cite{1,9,21,22}. Mobile navigation and route recommendation have become a hot topic in this research field \cite{2,10,11, 12,13,14,20,23,24,25}.

The MSR problem presented by Ge \emph{et al.} in \cite{2} is rather different from the traditional problems such as Shortest-Path problem \cite{15, 16}, Traveling-Salesman problem \cite{17} and Vehicle-Scheduling problem \cite{18}. Because for the shortest path computation problem, the source and destination nodes of an object are known in advance. However, for MSR problem, both of them are unknown. The traditional Traveling-Salesman Problem (TSP) gets a shortest path that includes all $N$ locations while MSR problem is to find a path that consists of a subset of given $N$ locations. In addition, the traditional Vehicle-Scheduling problem needs to determine a set of duties in advance while the pick-up routes (jobs) among several locations is uncertain for the MSR problem.

In \cite{2}, the authors focus on the MSR problem with a length constraint due to the high computational complexity of the unconstraint simple MSR problem. To reduce the search space, they proposed a route dominance based sequence pruning algorithm LCP. However, the proposed algorithm has difficulty in handling the problem with a large number of pick-up points. A novel skyline based algorithm SkyRoute is also introduced for searching the optimal route which can service multiple cabs online. However, the skyline query is inefficient in handling , since it is processed online.
%

Yuan \emph{et al.} proposed a probability model for detecting pick-up points \cite{4}. It finds a route with the biggest pick-up probability to the parking position constrained by a distance threshold instead of the minimal cost of the route and provides location recommendation service both for the cab drivers and for the people needing the taxi services. In contrast, the problem solved in \cite{20, 24} is different from the MSR problem which is to recommend a fastest route to a destination place with starting position and time constraints.

Powell \emph{et al.} \cite{3} proposed a grid-based approach to suggest profit locations for taxi drivers by constructing a spatio-temporal profitability map, on which, the nearby regions of the driver are scored according to the potential profit calculated by the historical data. However, this method only finds a parking place with the biggest profit in a local scope instead of a set of pick-up points with overall consideration.

Lu \emph{et al}. \cite{10} introduced a problem of finding optimal trip route with time constraint. They also proposed an efficient trip planning method considering the current position of a user. However, their method uses the score of attractions to measure the preference of a route.

\section{PROPOSED METHOD}

To address the computational challenge of the generalized MSR problem, we first identify the iterative property of the PTD function,which makes the incremental generation of the potential sequences possible and then propose the pruning principle, which uses the iterative property to efficiently reduce the search space.

\subsection{The Iterative Property Of The PTD Function}
As described in section 2, the PTD function gives a computable measure for the cost of a route. In the following, we study the property of the PTD function.

Actually, an iterative computational formula of the PTD function \cite{5} can be obtained without considering the driving distance beyond the last pick-up point of a driving route. For this purpose, we introduce the concept of PTD sub-function.

\newtheorem{definition}{Definition}
\begin{definition}
\label{def6}
(PTD Sub-function $F1$) Given a driving route \(\vec d = \left\langle {{c_0},{c_1},{c_2}, \ldots ,{c_L}} \right\rangle \), \({c_0}\) is its starting point and \(\vec r = \left\langle {{c_1},{c_2}, \ldots ,{c_L}} \right\rangle \) is its pick-up sequence. Let the distance sub-vector of ${D(\vec d )}$ be
\begin{equation*}
\small
\widetilde {D(\vec d)} = \left\langle \begin{array}{l}
{D_{{c_0},{c_1}}},({D_{{c_0},{c_1}}} + {D_{{c_1},{c_2}}}), \ldots ,\sum\limits_{i=1}^L D_{c_{i-1} ,c_i}
\end{array} \right\rangle
\end{equation*}
and the probability sub-vector of ${P(\vec d)} $ be
\begin{equation*}
\small
\widetilde {P(\vec d)} = \left\langle \begin{array}{l}
P({c_1}),\overline {P({c_1})}  \cdot P({c_2}), \ldots ,\prod\limits_{i=1}^{L-1} \overline {P({c_i})} \cdot P(c_L )
\end{array} \right\rangle .
\end{equation*}
The PTD sub-function \(F1\) of the driving route $\vec d$ is defined as
\begin{equation}
\label{equ2}
{F1(\vec d) = \widetilde {D(\vec d)} \cdot \widetilde {P(\vec d)}}.
\end{equation}
\end{definition}

Compared to the distance vector ${D(\vec d)}$ and probability vector ${P(\vec d)}$, the sub-vectors $\widetilde {D(\vec d)}$ and $\widetilde {P(\vec d)}$ of a driving route $\vec d = \left\langle c_0, \vec r \right\rangle$ only lack the last component respectively. Therefore, the PTD cost of a driving route $\vec d$ can be expressed using its PTD subfunction by the following equation

\begin{equation}
\label{equ3}
F(\vec d) = F1(\vec d) + {D_\infty} \cdot \prod\limits_{i=1}^L \overline {P({c_i})}.
\vspace{-0.2cm}
\end{equation}

In fact, we do not have the starting point of a cab in the stage of offline processing. For enhancing the online search efficiency, we pre-compute the costs of all the potential sequences. The involved concept of probability summation function $PE$ is introduced as follows.
\begin{definition}
\label{def7}
(Probability Summation Function $PE$) Let \(\vec r = \left\langle {{c_1},{c_2}, \ldots ,{c_L}} \right\rangle \) be a potential sequence with length $L$ and $\vec d = \left\langle c_0, \vec r \right\rangle$ be a driving route derived from $\vec r$. The probability summation of $\vec r$ is the sum of all the dimensions in the probability sub-vector $\widetilde {P(\vec d)}$, and it is given as
\begin{equation}
\small
{PE(\vec r)}  = P({c_1}) + \overline {P({c_1})}  \cdot P({c_2}) +  \ldots  + \prod\limits_{i=1}^{L-1} \overline{P({c_i})}  \cdot P({c_L}).
\end{equation}
\end{definition}

Since the sum of all the components in the probability vector ${P(\vec d)}$ is equal to 1, the value of the probability summation function $PE$ of $\vec r$ has the following property.
\begin{equation}
\label{equ5}
PE(\vec r) = PE({c_1},{c_2}, \ldots ,{c_L}) = 1 - \prod\limits_{i=1}^L \overline {P({c_i})}.
\end{equation}
%
%
%
%

The value of the function $PE$ can be calculated recursively. Given a potential sequence \(\vec r_1 = \left\langle {{c_1},{c_2}, \ldots ,{c_k}} \right\rangle \) and its postfix  sub-sequence \(\vec r_2 = \left\langle {{c_2},{c_3}, \ldots ,{c_k}} \right\rangle \), \( PE(\vec r_1) \) can be iteratively calculated by

\begin{equation}
\label{equ10}
{PE(\vec r_1) = {P(c_1)} + \overline {P(c_1 )} \cdot PE(\vec r_2)}.
\end{equation}

According to the above definitions, we can obtain the iterative computation theorem of the potential sequences as follows.
\begin{theorem}
\label{theorem4}
Let \(\vec r = \left\langle {{c_1},{c_2}, \ldots ,{c_L}} \right\rangle \) be a potential sequence with length $L$. The distance sub-vector of $\vec r$ is \\
\vspace{-0.1cm}
\begin{equation*}
\small
{\widetilde {D(\vec r)} = \left\langle {{D_{{c_1},{c_2}}},({D_{{c_1},{c_2}}} + {D_{{c_2},{c_3}}}), \ldots ,\sum\limits_{i=2}^L {D_{{c_{i-1}},{c_i}}}} \right\rangle },
\end{equation*}
and its probability sub-vector is
\begin{equation*}
\small
{\widetilde {P(\vec r)} = \left\langle {P({c_2}),\overline {P({c_2})}  \cdot P({c_3}), \ldots ,\prod\limits_{i=2}^{L-1} \overline{P({c_i})}  \cdot P({c_L})} \right\rangle}.
\end{equation*}
Then the PTD sub-function \(F1\) of \(\vec r\) is
\vspace{-0.2cm}
\begin{equation*}
\small
{F1(\vec r) = \widetilde{D(\vec r)} \cdot \widetilde{P(\vec r)}}.
\end{equation*}

Given a potential sequence \(\vec r_1 = \left\langle {{c_1},{c_2}, \ldots ,{c_k}} \right\rangle \) $1 \le k \le N$ and its postfix  sub-sequence \(\vec r_2 = \left\langle {{c_2},{c_3}, \ldots ,{c_k}} \right\rangle \), the $F1(\vec r_1)$ can be iteratively calculated as
\begin{equation}
\label{equ7}
{F1(\vec r_1) = \overline {P(c_2 )}  \cdot F1(\vec r_2) + D_{c_1 ,c_2 }  \cdot {PE(\vec r_2)}}.
\end{equation}

\end{theorem}

\begin{proof}
Based on the definition of the PTD sub-function $F1$, we have
\scriptsize
\[
\begin{array}{l}
 F1(\vec r_1 ) = \widetilde{D(\vec r_1 )} \cdot \widetilde{P(\vec r_1 )} \\
  = D_{c_1 ,c_2 }  \cdot P(c_2 ) + \left( {D_{c_1 ,c_2 }  + D_{c_2 ,c_3 } } \right) \cdot \overline {P(c_2 )}  \cdot P(c_3 ) +  \cdots  + \sum\limits_{i = 2}^k {D_{c_{i - 1} ,c_i }  \cdot \prod\limits_{i = 2}^{k - 1} {\overline {P(c_i )}  \cdot P(c_k )} }  \\
  = D_{c_2 ,c_3 }  \cdot \overline {P(c_2 )}  \cdot P(c_3 ) +  \cdots  + \sum\limits_{i = 3}^k {D_{c_{i - 1} ,c_i }  \cdot \prod\limits_{i = 2}^{k - 1} {\overline {P(c_i )}  \cdot P(c_k )} } + D_{c_1 ,c_2 }  \cdot \left( {P(c_2 ) + \overline {P(c_2 )}  \cdot P(c_3 ) +  \cdots  + \prod\limits_{i = 2}^{k - 1} {\overline {P(c_i )}  \cdot P(c_k )} } \right) \\
  = \overline {P(c_2 )}  \cdot \left( {D_{c_2 ,c_3 }  \cdot P(c_3 ) +  \cdots  + \sum\limits_{i = 3}^k {D_{c_{i - 1} ,c_i }  \cdot \prod\limits_{i = 3}^{k - 1} {\overline {P(c_i )}  \cdot P(c_k )} } } \right)  + D_{c_1 ,c_2 }  \cdot \left( {P(c_2 ) + \overline {P(c_2 )}  \cdot P(c_3 ) +  \cdots  + \prod\limits_{i = 2}^{k - 1} {\overline {P(c_i )}  \cdot P(c_k )} } \right) \\
  = \overline {P(c_2 )}  \cdot F1(\vec r_2 ) + D_{c_1 ,c_2 }  \cdot PE(\vec r_2 )
 \end{array}
\]
\end{proof}

According to the Formulas~\ref{equ10} and \ref{equ7}, we can get the backward recursive formula for calculating the PTD sub-function $F1$ of the potential sequence \(\vec r\).

\begin{center}
\begin{tabular}{|l|} \hline
\textbf{The initial value:}\\
\(\forall c \in C\), \(F1(c) = 0\), \(PE(c) = P(c)\)\\
\textbf{Iterative formula:}\\
\(F1({c_1}, {c_2}, \ldots , {c_L}) = \overline {P({c_2})} \cdot F1({c_2}, {c_3}, \ldots , {c_L}) \) \(+ {D_{{c_1},{c_2}}} \cdot PE({c_2}, {c_3}, \ldots, {c_L})\)\\
\(PE({c_1}, {c_2}, \ldots , {c_L}) = P({c_1}) + \overline {P({c_1})} \cdot PE({c_2}, \ldots , {c_L})\)\\
\hline
\end{tabular}
\end{center}

The recursive formula given above shows that the \(F1\) value of the potential sequence \(\vec r = \left\langle {{c_1},{c_2}, \ldots ,{c_L}} \right\rangle \) can be recursively calculated by the \(F1\) and $PE$ values of its postfix sub-sequence \(\vec r' = \left\langle {{c_2},{c_3}, \ldots ,{c_L}} \right\rangle \). In the stage of offline analysis, we only have the set of potential pick-up points \(C\), but the locations of the cabs are unknown. Therefore, we can construct short postfix sequences and then incrementally add new pick-up points ahead of them, and this will lead to longer potential sequences. Actually, if we want to recommend a driving route with length $L$, we need to generate all potential sequences with length $L$. Once we get the current location ${c_0}$ of a cab online, we can obtain the driving routes satisfying the length constraint by inserting the current location of the cab ${c_0}$ to the head of the potential sequences with length $L$ as the starting point.

The PTD sub-function of the driving route $\vec d = \left\langle c_0, \vec r \right\rangle$ can be calculated using the values of \(F1(\vec r)\) and $PE(\vec r)$ via
\begin{equation}
\label{equ5}
F1(\vec d) = \overline {P({c_1})} \cdot F1(\vec r) + D_{{c_0},{c_1}} \cdot PE(\vec r).
\end{equation}

By combining Formula~\ref{equ3} with Formula~\ref{equ5}, the PTD cost of the driving route $\vec d$ can be calculated using the formula
\begin{equation}
\footnotesize
\label{equ6}
\begin{array}{l}
F(\vec d) = \overline {P({c_1})} \cdot F1(\vec r) + {{\rm{D}}_{{c_0},{c_1}}} \cdot PE(\vec r) + {D_\infty} \cdot \prod_{i=1}^L \overline {P({c_i})}\\
\quad \quad \; \, = \overline {P({c_1})}  \cdot F1(\vec r) + {D_{{c_0},{c_1}}} \cdot PE(\vec r) + {D_\infty} \cdot (1 - PE(\vec r)).
\end{array}
\end{equation}

%

For the driving route \(\vec d = \left\langle {{c_0},{c_1},{c_2}, \ldots ,{c_L}} \right\rangle \) with length \(L\), we can efficiently calculate the value of the PTD sub-function $F1$ of its pick-up sequence \(\vec r = \left\langle {{c_1},{c_2}, \ldots ,{c_L}} \right\rangle \) in advance. When the current location \({c_0}{\rm{ }}\) of the cab is received online, we can calculate the PTD value of \(\vec d\) based on Formula~\ref{equ6}. Then we can recommend the driving route satisfying the length constraint with the minimum PTD cost to the user.

Using the iterative property of the PTD function, we give a recursive computational formula for the PTD cost as well as an incremental backward path growth method which can generate a potential sequence from its postfix sub-sequence. In this way, we do not have to calculate the PTD cost for each possible driving route from scratch, but recursively calculate it from the $F1$ and $PE$ values of its postfix sub-sequences. Therefore, the cost of calculating the PTD of the routes can be reduced significantly.

\section{SEQUENCE PRUNING}

In [2], Ge \emph{et al.} proposed a sequence pruning algorithm LCP based on route dominance. Let us briefly illustrate the principle of route dominance based pruning used in algorithm LCP. In Figure 2, two potential sequences with length three \(\vec r_1 = \left\langle {c_1, c_2, c_5 } \right\rangle\) and  \(\vec r_2 = \left\langle {c_1, c_4, c_5 } \right\rangle\) have the same source and destination pick-up points. The associated $DP$ vectors are defined as \(DP(\vec r_1) = \left\langle {{D_{{c_1},{c_2}}},\overline {P({c_2})} ,{D_{{c_2},{c_5}}},\overline {P({c_5})} } \right\rangle \) and \(DP(\vec r_2) = \left\langle {{D_{{c_1},{c_4}}},\overline {P({c_4})} ,{D_{{c_4},{c_5}}},\overline {P({c_5})} } \right\rangle \). Because \(\left( {{D_{{c_1},{c_2}}} \le {D_{{c_1},{c_4}}}} \right) \wedge \left( {\overline {P({c_2})}  \le \overline {P({c_4})} } \right) \wedge \left( {{D_{{c_2},{c_5}}} \le {D_{{c_4},{c_5}}}} \right) \wedge \left( {\overline {P({c_5})}  \le \overline {P({c_5})} } \right)\) and \(\left( {{D_{{c_1},{c_2}}} < {D_{{c_1},{c_4}}}} \right) \vee \left( {\overline {P({c_2})}  < \overline {P({c_4})} }\right) \vee \left( {{D_{{c_2},{c_5}}} < {D_{{c_4},{c_5}}}} \right) \vee \left( {\overline {P({c_5})}  < \overline {P({c_5})} } \right) \) are both valid, we can infer that \(\vec r_1\) dominates \(\vec r_2\). Thus, \(\vec r_2\) will be pruned in advance by the algorithm LCP.

In algorithms LCP, all possible potential sequences should be generated. Since a route being dominated by another route depends on the value of each dimension of the DP vector, the pruning effect is not high. If we can identify and remove some non-optimal potential sequences incrementally in the stage of sequence generation, the pruning effect would be improved. Along this line, we introduce the sequence pruning principle adopted in our method.

\begin{definition}
\label{def10}
(Sequence Precedence) Given two potential sequences \({\vec a} = \left\langle {{c_{{a_1}}}, \ldots ,{c_{{a_k}}}} \right\rangle \) and \({\vec b} = \left\langle {{c_{{b_1}}}, \ldots ,{c_{{b_k}}}} \right\rangle \) with equal length $k$ (\(1 \le k \le N\)), for a starting position \({c_0}\), we will get two driving routes \({\vec d_1} = \left\langle {{c_0},{c_{{a_1}}}, \ldots ,{c_{{a_k}}}} \right\rangle \) and \({\vec d_2} = \left\langle {{c_0},{c_{{b_1}}}, \ldots ,{c_{{b_k}}}} \right\rangle \) with equal length $k$ derived from \(\vec a\) and \(\vec b\) respectively. If \(F({\vec d_1}) < F({\vec d_2})\) holds for any possible \({c_0}\), then \(\vec a\) precedes \(\vec b\), and it is denoted as \(\vec a \prec \vec b\).
\end{definition}
\begin{figure}[!t]
\centering
\includegraphics[width=3.2in]{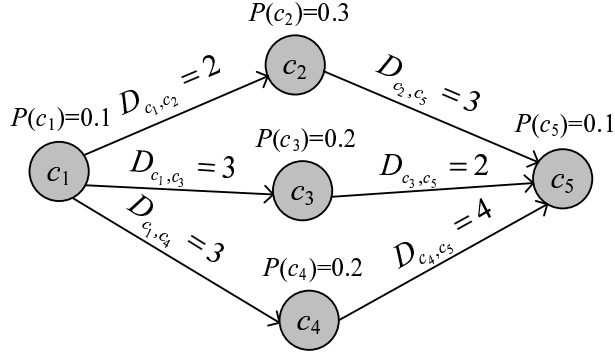}
\caption{An example of the sequence dominance and precedence.\label{fig2}}
\end{figure}

If \(\vec a \prec \vec b\), the potential sequence \(\vec b\) cannot form an optimal driving route and it should be removed from the collection of the sequence candidates in advance. For example, as shown in Figure~\ref{fig2}, there is another potential sequence $\vec r_3 = \left\langle {c_1, c_3 ,c_5} \right\rangle$. Since for any possible starting position \({c_0}\) the PTD value of the driving route $\vec d_1 = \left\langle {c_0, c_1 ,c_2 ,c_5 } \right\rangle $ must be smaller than that of the driving route $\vec d_3 = \left\langle {c_0, c_1 ,c_3 ,c_5 } \right\rangle $, \(\vec d_3\) is not an optimal driving route. Therefore, we can prune the potential sequence \(\vec r_3\) in the stage of offline processing.

It is easy to see that if \(\vec r\) dominates \(\vec r'\), then \(\vec r \prec \vec r'\) is also valid. On the contrary, if \(\vec r \prec \vec r'\), \(\vec r\) does not necessarily dominate \(\vec r'\). For example, as shown in Figure~\ref{fig2}, even though \(\vec r_1\) does not dominate \(\vec r_3\), \(\vec r_1 \prec \vec r_3\) is still valid. It shows that sequence dominance is only a special case of sequence precedence. As a result, the overall pruning effect of sequence precedence should be better than that of route dominance.

In order to efficiently evaluate the costs of the potential sequences, we provide a criterion of iterative precedence as follows.

\begin{definition}
\label{def11}
(Iterative Precedence) Let \(\vec a =  \left\langle {c_{{a_1}}}, \ldots ,{c_{{a_k}}} \right\rangle \) and \(\vec b =  \left\langle {c_{{b_1}}}, \ldots ,{c_{{b_k}}} \right\rangle \) be two potential sequences derived from the set of pick-up points $C$. If \(\left( {F1(\vec a) \le F1(\vec b)} \right) \wedge \left( {1 - PE(\vec a) < 1 - PE(\vec b)} \right)\) or \(\left( {F1(\vec a) < F1(\vec b)} \right) \wedge \left( {1 - PE(\vec a) \le 1 - PE(\vec b)} \right)\), then \(\vec a\) takes iterative precedence over \(\vec b\), denoted by \(\vec a \propto \vec b\).
\end{definition}

According to the definition of iterative precedence, we propose a method to determine the precedence relationship between pairs of potential sequences in order to prune some sequence candidates in the process of sequence generation. Note that since the PTD function $F$ and the iterative calculation of the PTD sub-function $F1$ are both relevant to the source point of the sequence, we only compare the PTD costs of the potential sequences with the same source point.

\begin{theorem}
\label{theorem5}
Let $1 \le k \le N-1$, $\vec a = \left\langle {{c_s},{c_{{a_1}}} \ldots ,{c_{{a_k}}}} \right\rangle$ and $\vec b = \left\langle {{c_s},{c_{{b_1}}} \ldots ,{c_{{b_k}}}} \right\rangle$ be two potential sequences with the same source point and the equal length $k+1$. \(\vec a' = \left\langle {{c_1},{c_2}, \ldots ,{c_m},{c_s},{c_{{a_1}}}, \ldots ,{c_{{a_k}}}} \right\rangle \) and \(\vec b' = \left\langle {{c_1},{c_2}, \ldots ,{c_m},{c_s},{c_{{b_1}}}, \ldots ,{c_{{b_k}}}} \right\rangle \) are two potential sequences derived from \(\vec a\) and \(\vec b\) by appending the same prefix sequence \(\vec r = \left\langle {c_1},{c_2}, \ldots ,{c_m} \right\rangle \) (\(0 \le m \le N - k - 1\), \({c_m} \in C\)) respectively. If \(\vec a \propto \vec b\), then \(\vec a' \prec \vec b'\).
\end{theorem}
\begin{proof}
Let \({c_0}\) be an arbitrary starting point, \(\overrightarrow {{d_a}}  = \left\langle {{c_0},{c_1},{c_2}, \ldots ,{c_m},{c_s},{c_{{a_1}}}, \ldots ,{c_{{a_{_k}}}}} \right\rangle \) and  \(\overrightarrow {{d_b}}  = \left\langle {{c_0},{c_1},{c_2}, \ldots ,{c_m},{c_s},{c_{{b_1}}}, \ldots ,{c_{{b_k}}}} \right\rangle \) be two driving routes associated with the potential sequences \(\vec a'\) and \(\vec b'\) respectively. \(\vec d = \left\langle {{c_0},{c_1},{c_2}, \ldots ,{c_m}} \right\rangle \) is a driving route with pick-up point sequence \(\vec r = \left\langle {c_1},{c_2}, \ldots ,{c_m} \right\rangle \).

Let ${D_0} = {D_{{c_0},{c_1}}} + {D_{{c_1},{c_2}}} + {D_{{c_2},{c_3}}} +  \ldots {\rm{ + }}{D_{{c_{m}},{c_s}}}$ and $\overline {{P_0}}  = \overline {P({c_1})}  \cdot \overline {P({c_2})}  \cdot \overline {P({c_3})} \cdot  \ldots  \cdot \overline {P({c_m})} $, then
\[
\begin{array}{l}
F{\rm{(}}\overrightarrow {{d_a}} {\rm{) = }}F1(\vec d){\rm{ + }}{D_0} \cdot \overline {{P_0}}  + ({D_\infty } - {D_0}) \cdot \overline {{P_0}}  \cdot \overline {P({c_s})}  \cdot \overline {P({c_{{a_1}}})} \cdot \overline {P({c_{{a_2}}})}  \cdot  \ldots  \cdot \overline {P({c_{{a_k}}})}  + F1(\vec a) \cdot \overline {P_0}  \cdot \overline {P({c_s})}, \\
F{\rm{(}}\overrightarrow {{d_b}} {\rm{) = }}F1(\vec d){\rm{ + }}{D_0} \cdot \overline {{P_0}}  + ({D_\infty } - {D_0}) \cdot \overline {{P_0}}  \cdot \overline {P({c_s})}  \cdot \overline {P({c_{{b_1}}})}  \cdot \overline {P({c_{{b_2}}})}  \cdot  \ldots  \cdot \overline {P({c_{{b_k}}})}  + F1(\vec b) \cdot \overline {{P_0}}  \cdot \overline {P({c_s})}.
\end{array}
\]

As we know,
\begin{equation*}
\overline {P({c_s})}  \cdot \overline {P({c_{{a_1}}})}  \cdot \overline {P({c_{{a_2}}})}  \cdot  \ldots  \cdot \overline {P({c_{{a_k}}})} =1 - PE(\vec a),
\end{equation*}
\begin{equation*}
\overline {P({c_s})}  \cdot \overline {P({c_{{b_1}}})}  \cdot \overline {P({c_{b2}})}  \cdot  \ldots  \cdot \overline {P({c_{{b_k}}})} = 1 - PE(\vec b).
\end{equation*}
Then
\begin{equation*}
\begin{array}{l}
F{\rm{(}}\overrightarrow {{d_a}} {\rm{) = }}F1(\vec d){\rm{ + }}{{\rm{D}}_0} \cdot \overline {{P_0}} {\rm{ + }}\overline {{P_0}}  \cdot (F1(\vec a) \cdot \overline {{\rm{P}}({c_s})} + ({D_\infty } - {D_0}) \cdot (1 - PE(\vec a))),\\
\end{array}
\end{equation*}

\begin{equation*}
\begin{array}{l}
F{\rm{(}}\overrightarrow {{d_b}} {\rm{) = }}F1(\vec d){\rm{ + }}{{\rm{D}}_0} \cdot \overline {{P_0}} {\rm{ + }}\overline {{P_0}}  \cdot (F1(\vec b) \cdot \overline {P({c_s})} + ({D_\infty } - {D_0}) \cdot (1 - PE(\vec b))).
 \end{array}
\end{equation*}

So,
\vspace{-0.1cm}
\begin{equation*}
\begin{array}{l}
 F{\rm{(}}\overrightarrow {d_a } {\rm{)}} - F{\rm{(}}\overrightarrow {d_b } {\rm{)}} = \overline {P_0 }  \cdot \overline {P(c_s )}  \cdot (F1(\vec a) - F1(\vec b)) +  \overline {P_0 }  \cdot (D_\infty   - D_0 ) \cdot (PE(\vec b) - PE(\vec a))). \\
 \end{array}
\end{equation*}
Since the desired travel distance increases along with the length of suggested driving routes, we can get \({D_\infty } > {D_0}\). Thus, if \(\vec a \propto \vec b\), i.e., \(\left( {F1(\vec a) \le F1(\vec b)} \right) \wedge \left( {1 - PE(\vec a) < 1 - PE(\vec b)} \right)\) or \(\left( {F1(\vec a) < F1(\vec b)} \right) \wedge \left( {1 - PE(\vec a) \le 1 - PE(\vec b)} \right)\), then \(F(\overrightarrow {{d_a}} ) < F(\overrightarrow {{d_b}} )\). That is to say, \(\vec a' \prec \vec b'\).
\end{proof}
According to the feature of the precedence relationship between the potential sequences, we introduce the theory of batch and incremental sequence pruning as follows.
\begin{corollary}
\label{corollary1}
(Batch Pruning) Given two potential sequences \(\vec a = \left\langle {{c_s},{c_{{a_1}}}, \ldots ,{c_{{a_k}}}} \right\rangle \) and \(\vec b = \left\langle {{c_s},{c_{{b_1}}}, \ldots ,{c_{{b_k}}}} \right\rangle \) with the equal length $k+1 (1 \le k \le N-1)$ and the same source point $c_s$, if \(\vec a \propto \vec b\), then \(\vec a \prec \vec b\).
\end{corollary}

The above corollary shows that if \(\vec a \propto \vec b\), the driving route $\left\langle {{c_0},\vec b} \right\rangle$ derived from the potential sequence \(\vec b\) is not an optimal driving route. Thus, \(\vec b\) should be pruned from the sequence candidates with length \(k + 1\).

\begin{figure}[!t]
\centering
\includegraphics[width=3.2in]{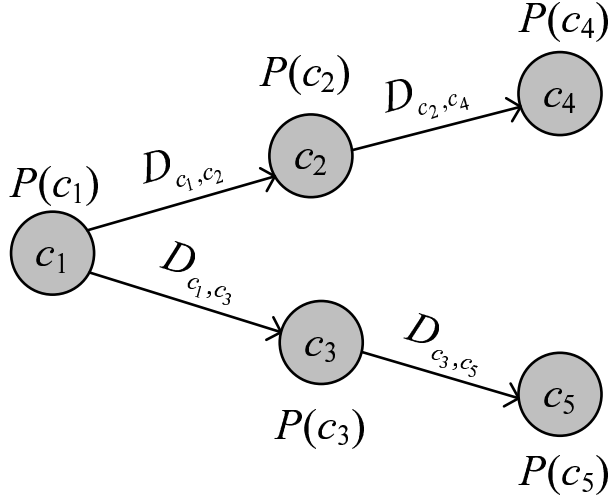}
\caption{An example of batch sequence pruning.}
\end{figure}

In the batch pruning, we can compare the PTD cost among potential sequences with length \( L(2 \le L \le N)\) using the values of $F1$ and $PE$ calculated in the iterative process. However, the batch pruning cannot be applied during the process of incremental sequence generation. For example, as shown in Figure 3, there are two potential sequences \(\vec r = \left\langle {c_1}, {c_2}, {c_4} \right\rangle\) and \(\vec r' = \left\langle {c_1}, {c_3}, {c_5} \right\rangle\). We can generate a new potential sequence \(\left\langle {c_3}, {c_1}, {c_2}, {c_4} \right\rangle\) considering \(\vec r\) as its postfix. However, we cannot append \({c_3}\) ahead of \(\vec r'\), because the pick-up point \({c_3}\) has existed in \(\vec r'\). Even though \(\vec r \propto \vec r'\), we can not prune \(\vec r'\) in advance in the process of the incremental backward path growth. For example, if \(\left\langle {c_2}, {c_1}, {c_3}, {c_5} \right\rangle \prec \left\langle {c_3}, {c_1}, {c_2}, {c_4} \right\rangle\), we may miss the optimal route for the improper pruning of \(\vec r'\) in advance. Along this line, we proposed a new corollary suitable for pruning potential sequences incrementally.

\begin{corollary}
\label{corollary2}
(Incremental Pruning) Given two potential sequences \(\vec a = \left\langle {{c_s},{c_{{a_1}}}, \ldots ,{c_{{a_k}}}} \right\rangle \) and \(\vec b = \left\langle {{c_s},{c_{{b_1}}}, \ldots ,{c_{{b_k}}}} \right\rangle \) with the equal length $k+1 (2 \le k \le N-1)$ and the same source point $c_s$, if \(\{ {c_{{a_1}}}, \ldots ,{c_{{a_{_k}}}}\}  = \{ {c_{{b_1}}}, \ldots ,{c_{{b_{_k}}}}\} \) and \(F1(\vec a) < F1(\vec b)\), then  all the driving routes having the postfix sub-sequence \(\vec b\) cannot be an optimal driving route.
\end{corollary}

As we know, the major obstacle why batch pruning is not suitable for incrementally pruning potential sequences is that we may not append the same prefix sequence for all of the potential sequences with the same source point and length. In Corollary~\ref{corollary2}, we add a constraint that the involved potential sequences must have the same source point and the same set of pick-up points. Then, it is obvious that all the involved potential sequences with the same length can be appended with the same possible prefix sequence. Since \(1 - PE(\vec a) = \overline {{P_{{c_s}}}}  \cdot \overline {{P_{{c_{{a_1}}}}}}  \cdot  \ldots  \cdot \overline {{P_{{c_{{a_k}}}}}} \), \(1 - PE(\vec b) = \overline {{P_{{c_s}}}}  \cdot \overline {{P_{{c_{{b_1}}}}}}  \cdot  \ldots  \cdot \overline {{P_{{c_{{b_k}}}}}} \) and \(\{ {c_{{a_1}}}, \ldots ,{c_{{a_{_k}}}}\}  = \{ {c_{{b_1}}}, \ldots ,{c_{{b_{_k}}}}\} \), then \(1 - PE(\vec a)\) = \(1 - PE(\vec b)\). Thus, we can simplify the iterative condition of precedence as \(F1(\vec a) < F1(\vec b)\).

\begin{figure}[!t]
\centering
\includegraphics[width=3.2in]{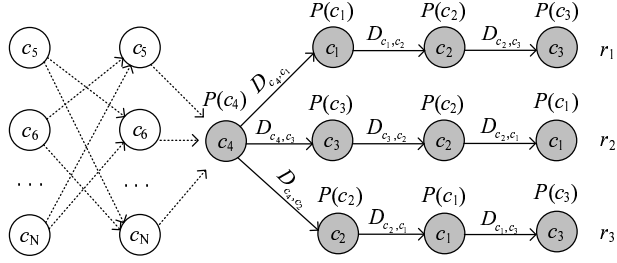}
\caption{An example of incremental sequence growing and pruning.\label{fig5}}
\end{figure}

Let us study the example shown in Figure~\ref{fig5}. There are three sequences with length 4: \({\vec r_1} = \left\langle {c_4}, {c_1}, {c_2}, {c_3}\right\rangle\) , \({\vec r_2} = \left\langle {c_4}, {c_3}, {c_2}, {c_1}\right\rangle\) and \({\vec r_3} = \left\langle {c_4}, {c_2}, {c_1}, {c_3} \right\rangle\). For any pick-up point \(c \in C - \{ {c_1},{c_2},{c_3},{c_4}\} \), it can be appended ahead of the three sequences to construct three new potential sequences with length 5. If \(F1({\vec r_1}) < F1({\vec r_2}) < F1({\vec r_3})\), then \({\vec r_1} \propto {\vec r_2} \propto {\vec r_3}\), i.e., \({\vec r_1} \prec {\vec r_2} \prec {\vec r_3}\). That is to say, \({\vec r_2}\) and \({\vec r_3}\) can be pruned in advance. Because any possible driving routes with a postfix sequence of the pruned sequence are not the optimal routes, they can be removed incrementally. However, \({\vec r_1}\) remains as a sequence candidate with length 4 and it is considered as the possible postfix of other longer potential sequences.

\subsection{The Analysis of Pruning Effect}
In this subsection, we analyze the pruning ratio of our incremental and batch pruning methods respectively. Let the total number of potential sequences be $M$ and the number of the remaining sequences after pruning be $M'$, the pruning ratio $\eta  = {{(M - M')} \mathord{\left/ {\vphantom {{(M - M')} M}} \right. \kern-\nulldelimiterspace} M}$.

\begin{theorem}
\label{theorem6}
For all possible potential sequences with length $L$ \((3 \le L \le N)\), the incremental pruning ratio is
\(\eta  = {{1 - 1} \mathord{\left/
 {\vphantom {{1 - 1} {(L - 1)!}}} \right.
 \kern-\nulldelimiterspace} {(L - 1)!}}\).
\end{theorem}
\begin{proof}
Given a set of potential pick-up points $C$ with $|C|=N$, the number of the potential sequences with length $L$ \((3 \le L \le N)\) is  \(M = \left( {\begin{array}{*{20}{c}}N\\L\end{array}} \right) \cdot L!\). In the process of incremental pruning, we only consider a group of potential sequences with the same source point and the same set of pick-up points. Based on Corollary~\ref{corollary2}, the precedence relationships of these sequences are only related to the $F1$ values of them. Actually, in most cases we choose an optimal sequence with the minimum $F1$ value from all these potential sequences. Since the number of the permutation of $L-1$ pick-up points except for the same source point is \((L - 1)!\), the number of the remaining sequences
\(M' = {{\left( {\begin{array}{*{20}c}
   N  \\
   L  \\
\end{array}} \right) \cdot L!} \mathord{\left/
 {\vphantom {{\left( {\begin{array}{*{20}c}
   N  \\
   L  \\
\end{array}} \right) \cdot L!} {(L - 1)!}}} \right.
 \kern-\nulldelimiterspace} {(L - 1)!}}  = \left( {\begin{array}{*{20}{c}}N\\L\end{array}} \right) \cdot L\).
Then the pruning percentage is \(
\eta  = {{(M - M')} \mathord{\left/
 {\vphantom {{(M - M')} {M = 1 - }}} \right.
 \kern-\nulldelimiterspace} {M = 1 - }}{{\left( {\left( {\begin{array}{*{20}c}
   N  \\
   L  \\
\end{array}} \right) \cdot L} \right)} \mathord{\left/
 {\vphantom {{\left( {\left( {\begin{array}{*{20}c}
   N  \\
   L  \\
\end{array}} \right) \cdot L} \right)} {\left( {\left( {\begin{array}{*{20}c}
   N  \\
   L  \\
\end{array}} \right) \cdot L!} \right)}}} \right.
 \kern-\nulldelimiterspace} {\left( {\left( {\begin{array}{*{20}c}
   N  \\
   L  \\
\end{array}} \right) \cdot L!} \right)}} = 1 - {1 \mathord{\left/
 {\vphantom {1 {(L - 1)!}}} \right.
 \kern-\nulldelimiterspace} {(L - 1)!}}
 .\)
\end{proof}

Note that the incremental pruning method is only applied to deal with the potential sequences with the length \(L \ge 3\). According to Theorem~\ref{theorem6}, the incremental pruning ratio sharply increases along with the increase of the length of the sequences. In order to remove more non-optimal sequences, we need to use the batch pruning method on the remaining sequences after the incremental pruning process. As a result, the pruning ratio can be improved further.

In the process of batch pruning, we compare the precedence relations between the remaining potential sequences with the same source point $c \in C$ and the same length $L$. As we know, whether a potential sequence will be removed by the batch pruning method is related to the $F1$ and $PE$ values of the sequences. Therefore, with the increase of the length $L$, the probability of the equivalence of the $F1$ and $PE$ for any pair of sequences with the same source point will become lower and lower. As a result, the number of the remaining sequence candidates after incremental and batch pruning process is close or equal to $N$ when the length $L$ is close to $N$.
%

\section{THE ALGORITHM}

Based on the analysis above, we first present the offline generation algorithm of the potential sequence candidates and the online route query algorithm. Then, we analyze the computational complexity of our algorithms.

\subsection{The Offline Processing Algorithms}
The detail of our dynamic programming based algorithm BP-Growth is given in Algorithm~\ref{algorithm1}. It generates the potential sequence candidates in the offline stage when the position of a cab is not involved. In order to construct all possible potential sequence candidates incrementally and efficiently, a backward path growth procedure and an incremental sequence pruning process are employed which combines with the iterative calculation of the $F1$ and $PE$ values of the potential sequences.

\begin{algorithm}[!htb]
\scriptsize
\caption{BP-Growth\label{algorithm1}}
\begin{algorithmic}[1]
\renewcommand{\algorithmicrequire}{\textbf{Input:}}
\renewcommand{\algorithmicensure}{\textbf{Output:}}
\REQUIRE{A set of potential pick-up points $C$, the  probability set $P$ for all pick-up points, the pairwise driving distance matrix $D$ of pick-up points.}
\ENSURE{A set of the potential sequence candidates \(\overrightarrow R \) with length $L$ from 1 to $N$}
\STATE \(\overrightarrow {{R^1}}  \leftarrow \emptyset \);
\FOR{each \({c_i} \in C\)}
    \STATE \(\vec r \leftarrow  \left\langle{c_i}\right\rangle\); \(F1(\vec r) \leftarrow 0\); \(PE(\vec r) \leftarrow P(c_i )\); \(\overrightarrow {R^1 }  \leftarrow \overrightarrow {R^1 }  \cup \{ \vec r\}\);
\ENDFOR
\FOR{$L=2$ to $N$}
    \STATE \(\overrightarrow {{R^L}}  \leftarrow \emptyset \);
        \FOR{each \(\vec r \in \overrightarrow {{R^{L - 1}}} \)}
            \FOR{each \({c_i} \in \left( {C - {C_{\vec r}}} \right)\)}
                \STATE //Potential Sequence Generation
                \STATE \(\vec p \leftarrow \left\langle{{c_i},\vec r}\right\rangle \); \(c \leftarrow s(\vec r)\);
                \STATE \(F1(\vec p) \leftarrow F1(\vec r) \cdot \overline {P(c)}  + {D_{{c_i},c}}PE(\vec r)\);
                \STATE \(PE(\vec p) \leftarrow PE(\vec r) \cdot \overline {P({c_i})}  + P({c_i})\);
                \STATE //Incremental Sequence Pruning
                \STATE \(\overrightarrow {R_{\vec p}^L}  = \{ \vec q|\vec q \in \overrightarrow {{R^L}} ,s(\vec q) = s(\vec p),{C_{\vec q}} = {C_{\vec p}}\}\);
                \IF{\(\overrightarrow {R_{\vec p}^L}  = \emptyset \)}
                    \STATE \(\overrightarrow {{R^L}}  \leftarrow \overrightarrow {{R^L}}  \cup \{ \vec p\} \);
                \ELSE
                    \IF{\(\forall \vec q \in \overrightarrow {{R_{\vec p}^L}} ({\rm{F}}1(\vec p) = F1(\vec q))\)}
                        \STATE \(\overrightarrow {{R^L}}  \leftarrow \overrightarrow {{R^L}}  \cup \{ \vec p\} \);
                    \ENDIF
                \ELSE
                    \IF{\(\forall \vec q \in \overrightarrow {{R_{\vec p}^L}} (F1(\vec p) < F1(\vec q))\)}
                        \STATE \(\overrightarrow {{R^L}}  \leftarrow \left( {\overrightarrow {{R^L}}  - \overrightarrow {R_{\vec p}^L} } \right) \cup \{ \vec p\} \);
                    \ENDIF
                \ENDIF
            \ENDFOR
        \ENDFOR
\ENDFOR
\RETURN {\(\overrightarrow R  = \mathop  \cup \limits_{L = 1}^{^{{N}}} \overrightarrow {{R^L}} \);}
\end{algorithmic}
\end{algorithm}

After the sequence generation and pruning process of Algorithm~\ref{algorithm1}, we will obtain a set of sequence candidates with length from 1 to $N$. For the potential sequence candidates, we adopt the batch pruning algorithm to reduce the number of sequence candidates further. As we know, after the sequence candidates are produced offline, the $F1$ and $PE$ values of these sequences have also been calculated iteratively. Therefore, we can directly compare the $F1$ and $PE$ values between the potential sequence candidates to prune the non-optimal ones during the batch pruning process which is described in Algorithm~\ref{algorithm2}.

\begin{algorithm}[!htb]
\scriptsize
\caption{BatchPruning\label{algorithm2}}
\begin{algorithmic}[1]
\renewcommand{\algorithmicrequire}{\textbf{Input:}}
\renewcommand{\algorithmicensure}{\textbf{Output:}}
\REQUIRE{A set of the potential sequences \(\overrightarrow {{R^L}} \) with length $L$.}
\ENSURE{A set of the remaining sequence candidates \(\overrightarrow {{{R'}^L}} \) with length $L$.}
\FOR {each \(c \in C\)}
    \STATE \(\overrightarrow {R_c^L}  \leftarrow \emptyset \);
\ENDFOR
\FOR {each \(\vec r \in \overrightarrow {{R^L}} \)}
    \STATE \(c \leftarrow s(\vec r)\);
    \STATE \(\overrightarrow {R_c^L}  \leftarrow \overrightarrow {R_c^L}  \cup \left\{ {\overrightarrow r } \right\}\);
        \FOR{each \(\overrightarrow q  \in \overrightarrow {R_c^L}  \wedge \overrightarrow r  \ne \overrightarrow q \)}
            \IF{\(\vec q \propto \vec r\)}
                \STATE \(\overrightarrow {R_c^L}  \leftarrow \overrightarrow {R_c^L}  - \left\{ {\overrightarrow r } \right\}\);
                \STATE break;
            \ELSE
                \IF{\(\vec r \propto \vec q\)}
                    \STATE \(\overrightarrow {R_c^L}  \leftarrow \overrightarrow {R_c^L}  - \left\{ {\overrightarrow q } \right\}\);
                \ENDIF
            \ENDIF
        \ENDFOR
\ENDFOR
\RETURN{\(\overrightarrow {{R^{'L}}}  = \mathop  \cup \limits_{c \in C} \overrightarrow {R_c^L} \);}
\end{algorithmic}
\end{algorithm}

\subsection{The Online Search Algorithm}
Our method is able to provide real-time driving route recommendation service for the unloaded cabs at various positions. When a cab at the position \({c_0}\) requests the recommendation service, an online search algorithm is adopted to find an optimal driving route from the remaining potential sequences generated in the offline stage. Algorithm~\ref{algorithm3} shows the online search procedure of optimal route in detail.

\begin{algorithm}[!htb]
\scriptsize
\caption{RouteOnline\label{algorithm3}}
\begin{algorithmic}[1]
\renewcommand{\algorithmicrequire}{\textbf{Input:}}
\renewcommand{\algorithmicensure}{\textbf{Output:}}
\REQUIRE{: A set of the sequence candidates \(\overrightarrow R\), the current position of a cab \({c_0}\) and the minimum length $L_{\min}$ and maximum length $L_{\max}$ of the suggested driving route $\left( {1 \le L_{\min }  \le L_{\max }  \le N} \right)$.}
\ENSURE{: A set of the optimal driving routes \(\overrightarrow {D_{\min }} \).}
\STATE \(\overrightarrow {D_{\min }}  \leftarrow \emptyset \); \({F_{\min }} \leftarrow +\infty \);
\FOR{$L=L_{\min}$ to $L_{\max}$}
  \FOR{ each \(\vec r \in \overrightarrow {R^L } \)}
    \STATE \(c = s(\vec r)\);
    \STATE \(\vec d = \left\langle{{c_0},\vec r}\right\rangle \);
    \STATE \(F(\vec d) = F1(\vec r) \cdot \left( {1 - P(c)} \right) + {D_{{c_0},c}} \cdot PE(\vec r) + {D_\infty} \cdot (1 - PE(\vec r))\);
    \IF{\(\overrightarrow {D_{\min }}  = \emptyset  \vee F(\vec d) = {F_{\min }}\)}
        \STATE \(\overrightarrow {D_{\min }}  \leftarrow \overrightarrow {D_{\min }}  \cup \{\mathord{\buildrel{\lower3pt\hbox{$\scriptscriptstyle\rightarrow$}}\over d} \} \);
    \ELSE
        \IF{\(F(\vec d) < {F_{\min }}\)}
            \STATE \(\overrightarrow {D_{\min }}  \leftarrow \{ \mathord{\buildrel{\lower3pt\hbox{$\scriptscriptstyle\rightarrow$}}\over d} \} \);
                    \({F_{\min }} \leftarrow F(\vec d)\);
        \ENDIF
    \ENDIF
  \ENDFOR
\ENDFOR
\RETURN{\(\overrightarrow {D_{\min }} \);}
\end{algorithmic}
\end{algorithm}

In Algorithm~\ref{algorithm3}, for each $L\left( {L_{\min }  \le L \le L_{\max } } \right)$, we first generate the potential driving routes \(\overrightarrow {{D^L}} \) with length $L$ by connecting \({c_0}\) with each potential sequence candidate in the set \(\overrightarrow {{R^L}} \). Then we calculate the PTD value of each potential driving route with Formula~\ref{equ6}. Finally, the driving routes with the minimal PTD value are selected and returned to the users.

\subsection{Analysis of Computational Complexity}
In this subsection, we analyze the computational complexities of the offline sequence generation and the online search algorithm respectively.

\subsubsection{Offline processing algorithms}
We first analyze the computational complexity of our offline algorithm BP-Growth. The key step in the algorithm BP-Growth is the incremental process of the sequence growing and pruning. As we know, in order to generate the potential sequences with length $L$, we append each pick-up point $c$ ahead of the sequence candidates with length $L-1$ that do not contain $c$. When the length of the potential sequence $L=1$, all pick-up points will be enumerated, so the computational complexity is $N$. When $L=2$, as we know, the number of sequence candidates with length 1 is $N$. Since each pick-up point only appears once in a potential sequence, we still have $N-1$ possible pick-up points for each sequence candidate. Therefore, the loop execution times of the key step for $L=2$ is $N(N-1)$. When we generate the potential sequences with length $L>2$, the number of the remaining sequence candidates with length $L-1$ after the incremental pruning process is \({{\left( {\begin{array}{*{20}c}
   N  \\
   {L - 1}  \\
\end{array}} \right) \cdot (L - 1)!} \mathord{\left/
 {\vphantom {{\left( {\begin{array}{*{20}c}
   N  \\
   {L - 1}  \\
\end{array}} \right) \cdot (L - 1)!} {(L - 2)!}}} \right.
 \kern-\nulldelimiterspace} {(L - 2)!}}\). Nevertheless, we still have $N-L+1$ pick-up points to be appended to the heads of these sequence candidates, and the computational times is \( \left( {{\left( {\begin{array}{*{20}c}
   N  \\
   {L - 1}  \\
\end{array}} \right) \cdot (L - 1)!} \mathord{\left/
 {\vphantom {{\left( {\begin{array}{*{20}c}
   N  \\
   {L - 1}  \\
\end{array}} \right) \cdot (L - 1)!} {(L - 2)!}}} \right.
 \kern-\nulldelimiterspace} {(L - 2)!}}
 \right) \cdot (N - L + 1) = \left( {\begin{array}{*{20}{c}}N\\L\end{array}} \right) \cdot L \cdot (L - 1)\). It can be seen that the computational complexity of the process with length $L =1$ is $O(N)$. It increases gradually and reaches the peak with $L=\left\lfloor {{\raise0.7ex\hbox{$N$} \!\mathord{\left/
 {\vphantom {N 2}}\right.\kern-\nulldelimiterspace}
\!\lower0.7ex\hbox{$2$}}} \right\rfloor$. After that, the computational complexities decrease and drop to \(O({N^2})\) with $L =N$.

We then present the computational complexity analysis of our algorithm BP-Growth for generating all possible sequences with length from 1 to $N$.

Given a set of pick-up points $C$ with \(\left| C \right| = N\), as we know, the total execution times for generating all the potential sequences with length \(L \le N\) is \\
\(f(N) = N + N \cdot (N - 1) + \sum\limits_{L = 3}^N {(L - 1)}  \cdot L\left( {\begin{array}{*{20}{c}}N\\L\end{array}} \right)\).\\
 $f(N)$ can be transformed to
\[
\small
\begin{array}{l}
f(N) = N + 2\left( {\begin{array}{*{20}{c}}N\\2\end{array}} \right) + 2 \cdot 3\left( {\begin{array}{*{20}{c}}N\\3\end{array}} \right) +  \ldots  + (N - 2) \cdot (N - 1)\left( {\begin{array}{*{20}{c}}N\\{N - 1}\end{array}} \right) + (N - 1) \cdot N\left( {\begin{array}{*{20}{c}}N\\N\end{array}} \right)\\
\end{array}.
\]

Since \(L \cdot \left( {\begin{array}{*{20}{c}}N\\L\end{array}} \right) = (N - L + 1)\left( {\begin{array}{*{20}{c}}N\\{L - 1}\end{array}} \right)\), then \(f(N)\) can also be described by the following equation
\[
\small
\begin{array}{l}
f(N) = N + (N - 1)\left( {\begin{array}{*{20}{c}}N\\1\end{array}} \right) + 2(N - 2)\left( {\begin{array}{*{20}{c}}N\\2\end{array}} \right) + 3(N - 3)\left( {\begin{array}{*{20}{c}}N\\3\end{array}} \right) +  \ldots  + (N - 1)\left( {\begin{array}{*{20}{c}}N\\{N - 1}\end{array}} \right)
\end{array}.
\]
If we add above two equations, we will obtain the following deduction.
\[
\small
\begin{array}{l}
2f(N) = 2N + (N - 1)\left( {\begin{array}{*{20}{c}}N\\1\end{array}} \right) + 2(N - 1)\left( {\begin{array}{*{20}{c}}N\\2\end{array}} \right) +  3(N - 1)\left( {\begin{array}{*{20}{c}}N\\3\end{array}} \right) +  \ldots  + N(N - 1)\left( {\begin{array}{*{20}{c}}N\\N\end{array}} \right) \\
 \qquad\quad = 2N + (N - 1)(\left( {\begin{array}{*{20}{c}}N\\1\end{array}} \right) + 2\left( {\begin{array}{*{20}{c}}N\\2\end{array}} \right) + 3\left( {\begin{array}{*{20}{c}}N\\3\end{array}} \right) +  \ldots  + N\left( {\begin{array}{*{20}{c}}N\\N\end{array}} \right))\\
 \qquad\quad = 2N + N(N - 1) \cdot {2^{N - 1}}\\
\end{array}.
\]
Then \(f(N) = N + N(N - 1) \cdot {2^{N - 2}}\). Thus, \(O(f(N)) = O({N^2} \cdot {2^N})\).

In summary, the computational complexity of incremental generation of the potential sequences with all possible length $L \left( 1 \le L \le N \right)$ via BP-Growth is \(O({N^2} \cdot {2^N})\).

\subsubsection{Online search algorithm}
The computational complexity of our online search algorithm with $L_{\min} = L_{\max} = L$ directly depends on the number of the remaining sequence candidates in \(\overrightarrow {{R^L}} \). For the set of the sequence candidates \(\overrightarrow {{R^L}} \) produced by algorithm BP-Growth with incremental pruning, \(\left| {\overrightarrow {{R^L}} } \right| = \left( {\begin{array}{*{20}{c}}
N\\L\end{array}} \right) \cdot L  \). Therefore, when we set $L=1$ or $N$, the computational complexity of our online search algorithm RouteOnline is \(O\left( N \right)\). When \(L = \left\lfloor {{\raise0.7ex\hbox{$N$} \!\mathord{\left/
 {\vphantom {N 2}}\right.\kern-\nulldelimiterspace}
\!\lower0.7ex\hbox{$2$}}} \right\rfloor \), the computational complexity is the highest which is close to \(O\left( {N \cdot {2^N}} \right)\). If we use both the incremental and the batch pruning processes, the search efficiency can be significantly enhanced. However, it is hard to obtain the precise analysis of its computational complexity. As for the search time of route query with a constraint of minimum length $L_{\min}$ and maximum length $L_{\max}$, it is just the sum of the search time in each set of sequence candidates \(\overrightarrow {{R^L}} \left( L_{\min} \le L \le L_{\max} \right) \).

\section{EXPERIMENTAL EVALUATIONS}
In this section, we evaluate the performance of our method by comparing its pruning effect, Memory consumption and online search time with those of other state-of-the-art methods. All acronyms of evaluated algorithms are given in Table~\ref{table2}. LCP and SkyRoute are two route dominance based pruning algorithms proposed in \cite{2}. In particular, SkyRoute is an online pruning algorithm, where two skyline computing methods BNL and D\&C can be applied to prune potential sequences \cite{8}. Its corresponding online search methods are denoted by SR(BNL)S and SR(D\&C)S, respectively. All the algorithms were implemented in Visual C++ 6.0. The experiments were conducted on a PC with a Intel Pentium Dual E2180 processor and 4GB RAM.

\subsection{Data Sets}
The adopted experimental data sets are divided into two categories: real-world data and synthetic data.

\textbf{Real-World Data.} In the experiments, we adopt real-world cab mobility traces used in \cite{2}, which are provided
by Exploratorium - the museum of science, art and human perception. It contains GPS location traces of 514 taxis collected around 30 days in the San Francisco Bay Area. We extract 21,980 and 38,280 historical pick-up locations of all the taxi drivers on two time periods: 2PM-3PM and 6PM-7PM. In total, we obtain 10 and 25 clusters as well as their probabilities on these two real data sets using the same method adopted in \cite{2}.

\textbf{Synthetic Data.} We also generate four synthetic data sets. Specifically, we
randomly generate potential pick-up points and their pick-up probabilities within a special area by a standard uniform
distribution. In total, we have four synthetic data sets with 10, 15, 20 and 25 pick-up points respectively. The Euclidean distance instead of the driving distance is adopted to measure the distances between pairs of pick-up points.

For both real-world and synthetic data, we randomly generate the positions of the target cab for recommendation.

\begin{table*}[t]
\centering
\caption{Some acronyms used in experimental analysis. \label{table2}}
\begin{tabular}{|c|c|} \hline
LCP & Sequence pruning via route dominance\\
SkyRoute & Sequence pruning via skyline query\\
SR(BNL) & SkyRoute with skyline computing method BNL\\
SR(D\&C) & SkyRoute with skyline computing method D\&C\\
IP & Generation of potential sequence candidates\\
&via BP-Growth with incremental pruning \\
IBP & Generation of potential sequence candidates\\
& via BP-Growth with Incremental and Batch Pruning \\
\hline
BFS & Brute-force search \\
LCPS & Search via LCP \\
SR(BNL)S & Skyline search via the algorithm SkyRoute + BNL\\
SR(D\&C)S &  Skyline search via the algorithm SkyRoute + D\&C\\
IPS & Search on the potential sequences generated by IP \\
IBPS & Search on the potential sequences generated by IBP \\ \hline
\end{tabular}
\end{table*}

\subsection{The Overall Comparison of Pruning Effect}
As we know, algorithms BFS and LCP need to enumerate all possible sequences of a certain length $L$. For a set of potential pick-up points $C$ with $|C| =N$, the number of all possible sequences with length $L$ is \(\left( {\begin{array}{*{20}c}
   N  \\
   L  \\
\end{array}} \right) \cdot L!
\). And the computational complexity is \(O(N!)\) when $L=N$. When the number of pick-up points $N$ or the length of suggested route $L$ is a little larger (e.g., $N=20$ and $L=6$), both BFS and LCPS cannot finish the enumeration process in a rather long time. Therefore, when we analyze the pruning ratio varying with the length of suggested driving routes on the same set of pick-up points, we make the number of pick-up points small (e.g., $|C|=10$) in order to show the overall comparison of all concerned algorithms. When we analyze the pruning ratio varying with the number of pick-up points on the fixed length of driving routes, we also make the length of the routes small (e.g., $L=3$ and $L=5$). For the algorithms proposed in this paper, since the sequences are pruned incrementally, both the time and space complexity are better than that of BFS and LCP. Thus, we can use the synthetic data set with $|C|=25$ to analyze the pruning effect of the proposed incremental algorithm BP-Growth in detail.

\subsubsection{The Pruning Ratio Varying with the Length of Potential Sequence}

\begin{figure}[!htp]
\centering
\subfigure[Real-World Data]{\includegraphics[width=0.36\textwidth] {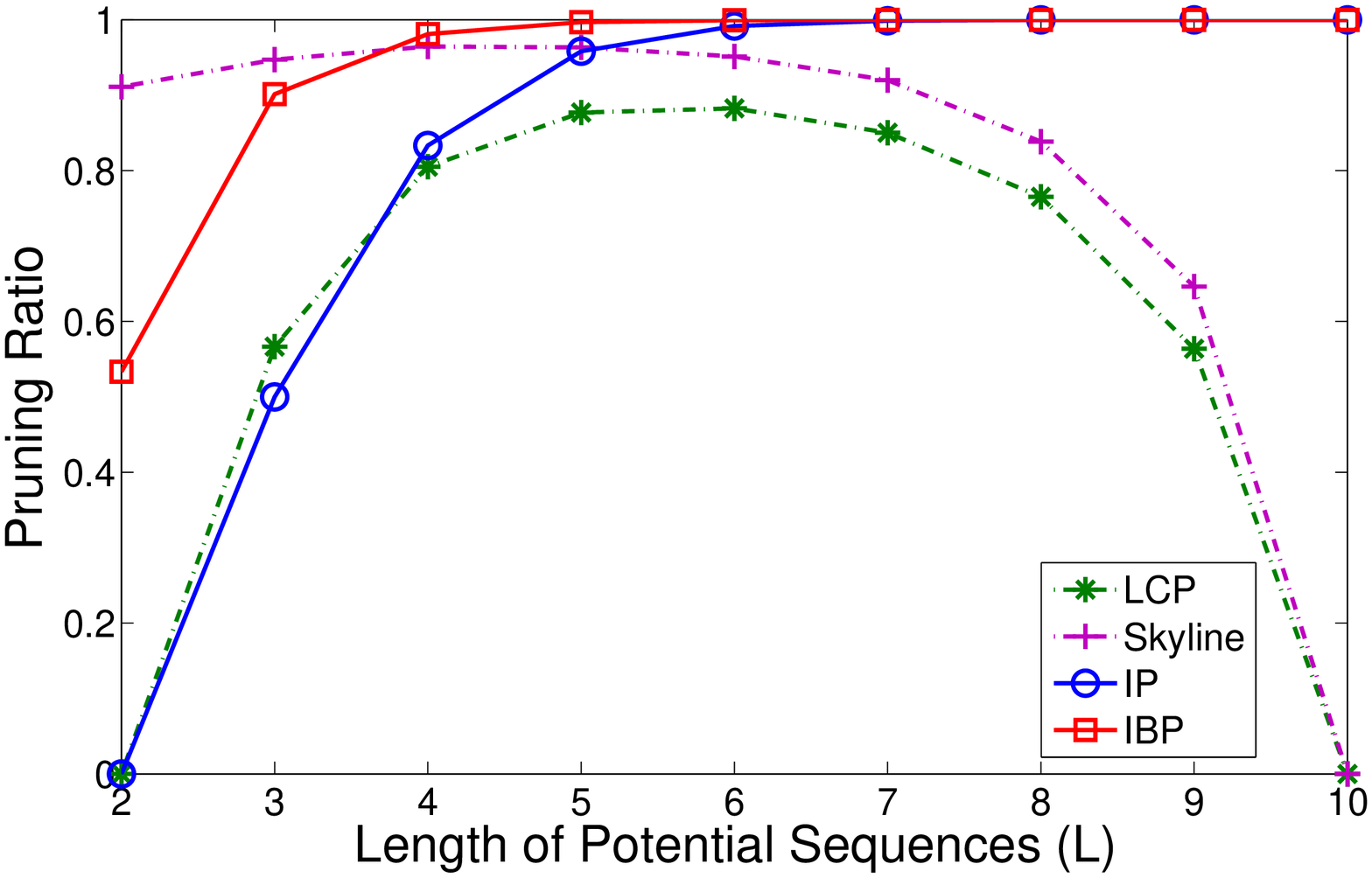}\label{pruneRealdataC10}}
\subfigure[Synthetic Data]{\includegraphics[width=0.36\textwidth] {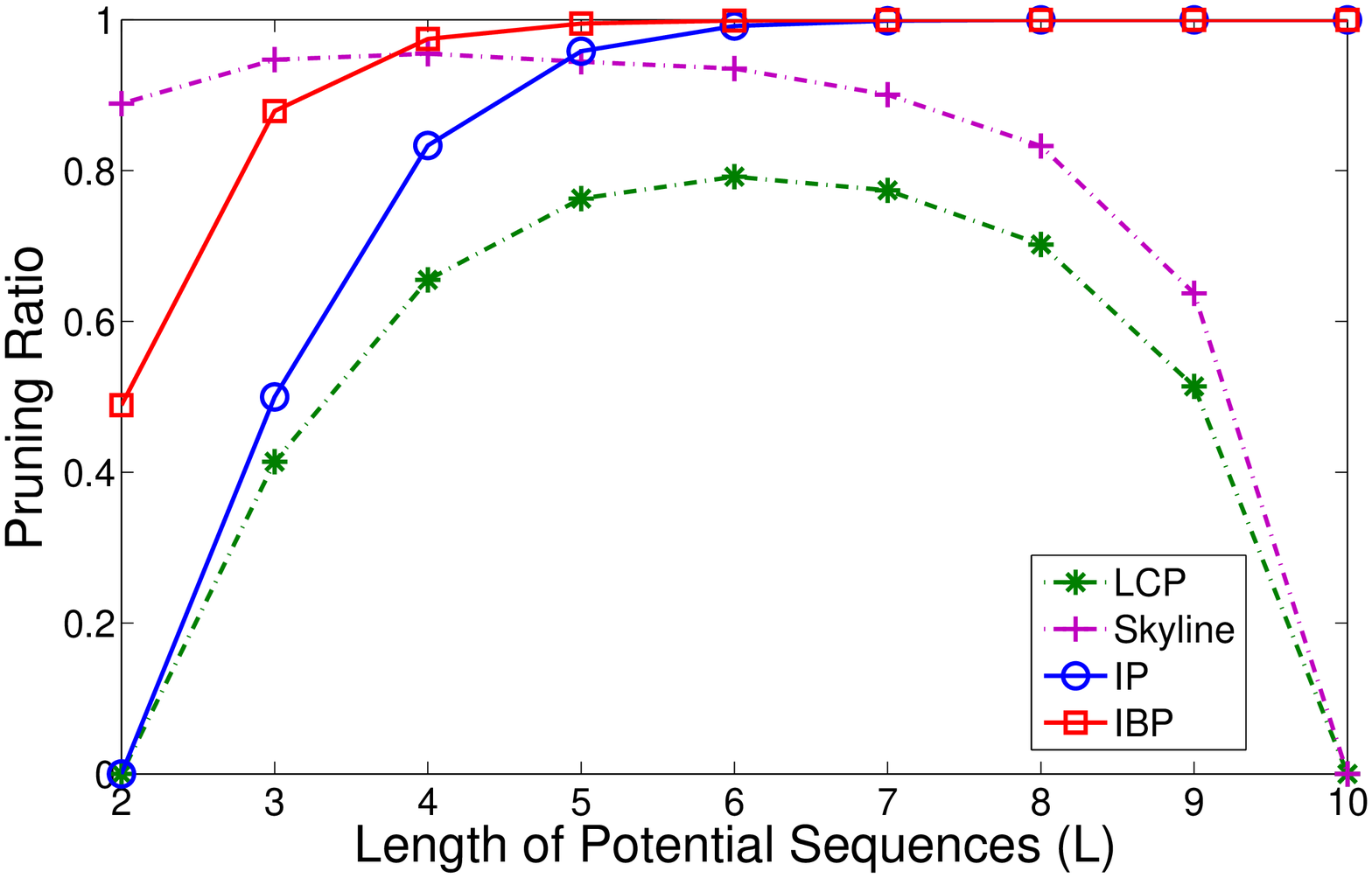}\label{pruneC10}}
\caption{The pruning ratio of different algorithms w.r.t. the length of potential sequence on the data sets with $|C|=10$.\label{pruneC}}
\end{figure}

Figure~\ref{pruneC} shows the varying of pruning ratio of several algorithms with the length of potential sequence on both real-world and synthetic data with $|C|=10$. Algorithms LCP, SkyRoute, IP and IBP are all able to prune some non-optimal sequences derived from $C$. When the length $L=2$ or $L=3$, the proposed algorithms IBP and IP perform worse than the algorithms SkyRoute and IP. However, as the length of the potential sequence $L$ increases, the pruning ratios of our algorithms IP and IBP are both significantly improved. It can be observed that IBP outperforms SkyRoute and IP outperforms LCP on both real and synthetic data when \(L \ge 5\). Furthermore, the pruning ratios of our algorithms are gradually improved and close to 1 when the length of suggested driving route \(L \ge 6\). In contrast, the change of the pruning ratios of LCP shows a trend of parabola. When \(L > 5\), the pruning ratios of LCP and SkyRoute both gradually drop. When the length is equal to the number of pick-up points (i.e., $L=|C|$), the pruning ratios of them decrease to 0.

%

To verify that our method can process the potential sequences derived from a larger number of pick-up points, we test the pruning ratios of IP and IBP on the synthetic data set with $|C|=25$. We find that the trends of the pruning ratios of our algorithms on different data sets are consistent. Since LCP and SkyRoute are only able to deal with the driving routes with \(L \le 5\) on the data set with $|C|=25$, we can not obtain the whole result of them on this bigger data set.

Let us analyze the reason why the pruning ratios of algorithms IP and IBP are relatively high. First, for the incremental pruning algorithm IP, its pruning ratio is equal to \({{1 - 1} \mathord{\left/
 {\vphantom {{1 - 1} {(L - 1)!}}} \right.
 \kern-\nulldelimiterspace} {(L - 1)!}}\) which dramatically increases along with the increase of the length of potential sequence. When $L=6$, the pruning percentage has reached 99.2\%.

\begin{figure}[!htp]
\centering
\includegraphics[width=0.8\textwidth]{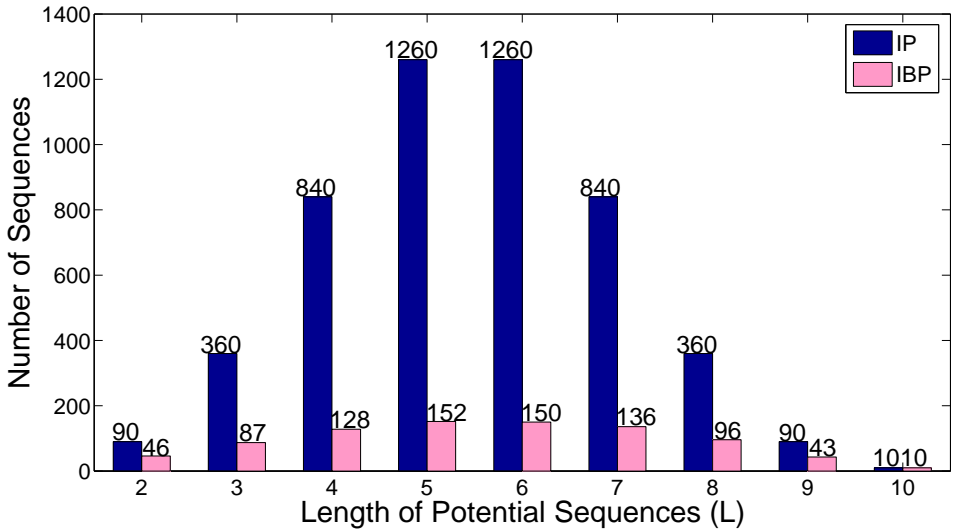}
\caption{The number of remaining sequence candidates after using the pruning algorithms IP and IBP respectively on the synthetic data set with $|C|=10$.\label{candidateC10}}
\end{figure}

\begin{figure}[!htp]
\centering
\subfigure[]{\includegraphics[width=0.48\textwidth] {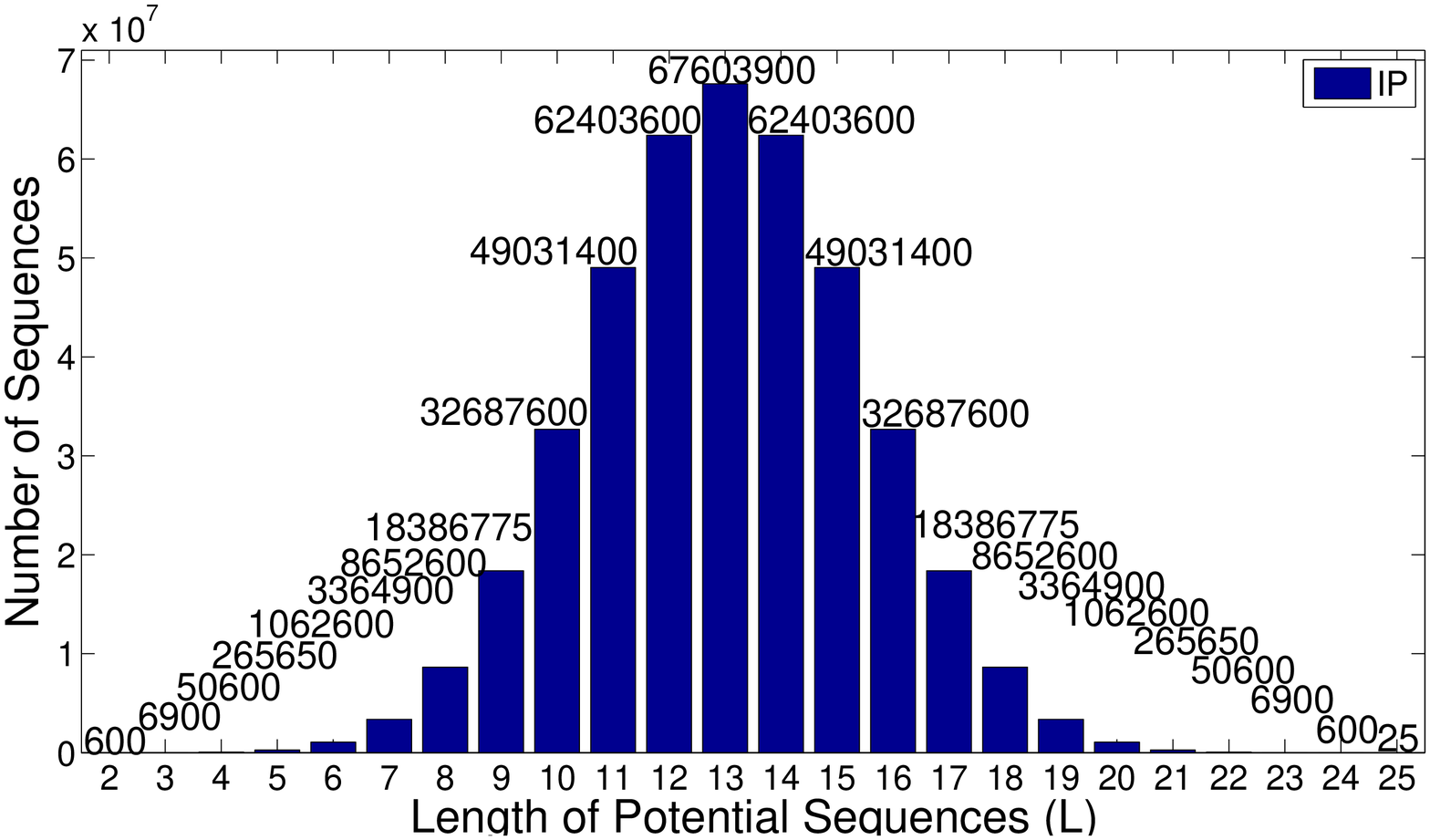}\label{candidateC25IP}}
\subfigure[]{\includegraphics[width=0.45\textwidth] {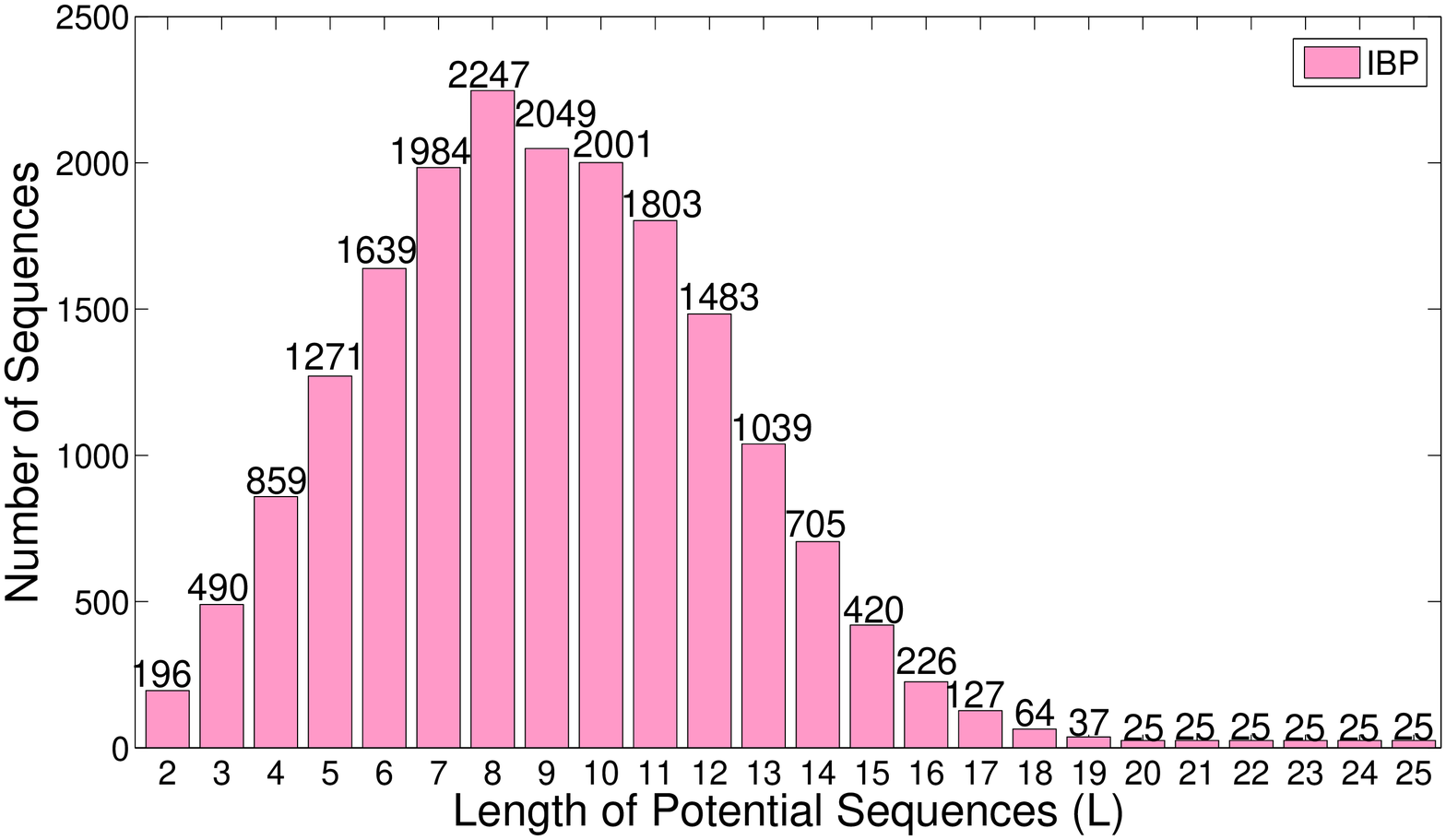}\label{candidateC25IBP}}
\caption{The number of remaining sequence candidates after using the pruning algorithms IP and IBP respectively on the synthetic data set with $|C|=25$.\label{candidateC25}}
\end{figure}

In addition, the algorithm IBP also adopts a batch pruning process to remove some non-optimal potential sequences. As shown in Figures~\ref{candidateC10} and \ref{candidateC25}, the number of remaining sequences of IBP can be several orders of magnitude smaller than that of IP, especially when $|C|$ is large, which demonstrates the effectiveness of batch pruning. The overall trend of the number of the remaining sequence candidates presents a Gaussian distribution. It increases first with the increase of the length $L$, and then decreases when $L \ge {{\left| C \right|} \mathord{\left/
 {\vphantom {{\left| C \right|} 2}} \right.
 \kern-\nulldelimiterspace} 2}$. Moreover, it is close to the number of pick-up points $|C|$ when $L \to |C|$, which is completely consistent with the analysis of Section 3.

In terms of LCP, whether a route is dominated by another route depends on the value of each dimension of the vector DP. When the value of $L$ is small, the pruning ratio has some growth with the increase of the length. However, when the sequence length becomes larger, the number of the dimensions of vector DP increases and the probability of domination in each dimension between DP vectors becomes lower, which leads to the gradual decline of the pruning ratio. When $L=|C|$, since all pick-up points are involved, it is impossible to make the probability of each pick-up point in a sequence larger than that of another. Thus, the pruning percentage is 0 in this case. As for SkyRoute, since the principle of it is similar as that of LCP, the overall trends of them are almost the same.

\subsubsection{The Pruning Ratio Varying with the Number of Pick-up Points}
\begin{figure}[!htp]
\centering
\subfigure[$L=3$]{\includegraphics[width=0.45\textwidth] {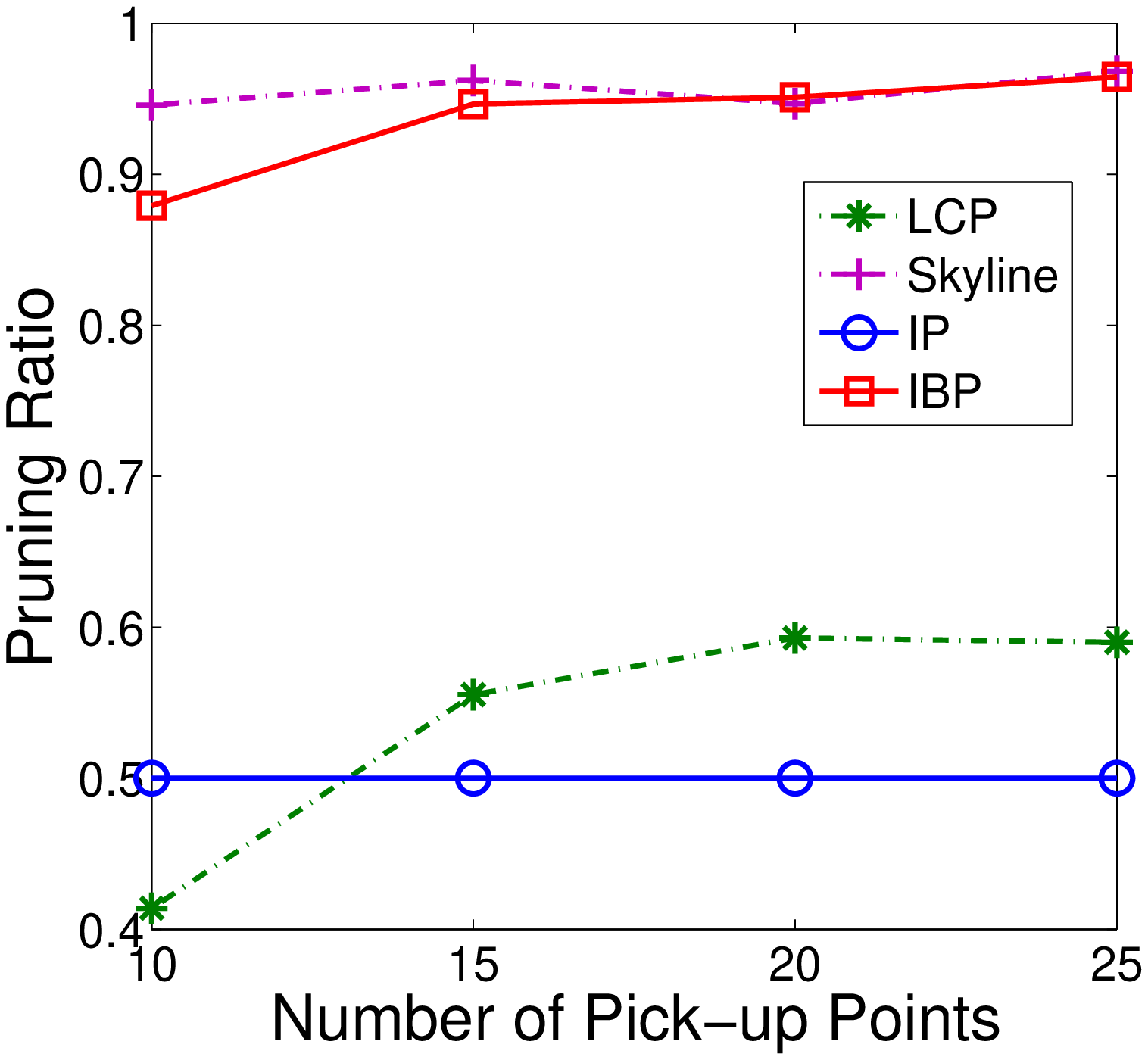}\label{pruneL3}}
\subfigure[$L=5$]{\includegraphics[width=0.45\textwidth] {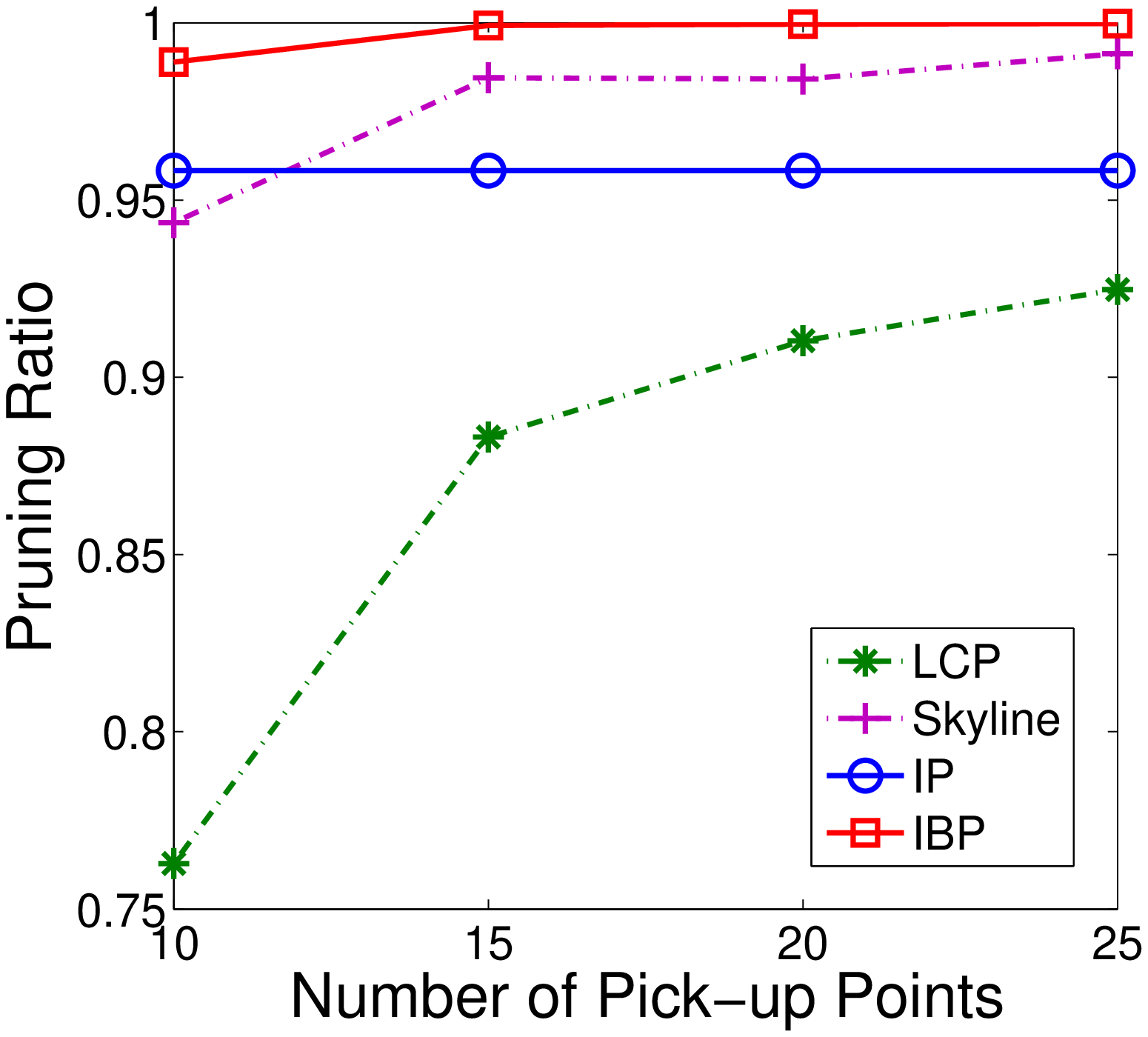}\label{pruneL5}}
\caption{The pruning ratio of different algorithms varying with the number of pick-up points on the synthetic data with $|C|=25$. \label{pruneL}}
\end{figure}
Figure~\ref{pruneL} shows the varying of pruning ratio with the number of pick-up points on synthetic data with $|C|=25$ when $L=3$ and $L=5$ respectively. It can be observed that the pruning ratios of LCP and IBP increase with the increase of the number of pick-up points, and the pruning ratio of IP is constant. When $L=3$, the pruning ratio of IP is equal to 0.5 and the pruning ratio of IBP gradually increases with the number of pick-up points. However, our algorithms do not perform better than algorithms LCP and SkyRoute. When $L=5$, the pruning ratio of IP is more than 0.95 and the pruning ratio of IBP is close to 1 which are much higher than those of LCP and SkyRoute respectively.

\subsection{Analysis of the Memory Consumption}
\begin{figure}[!htp]
\centering
\subfigure[$|C|=10$]{\includegraphics[width=0.45\textwidth] {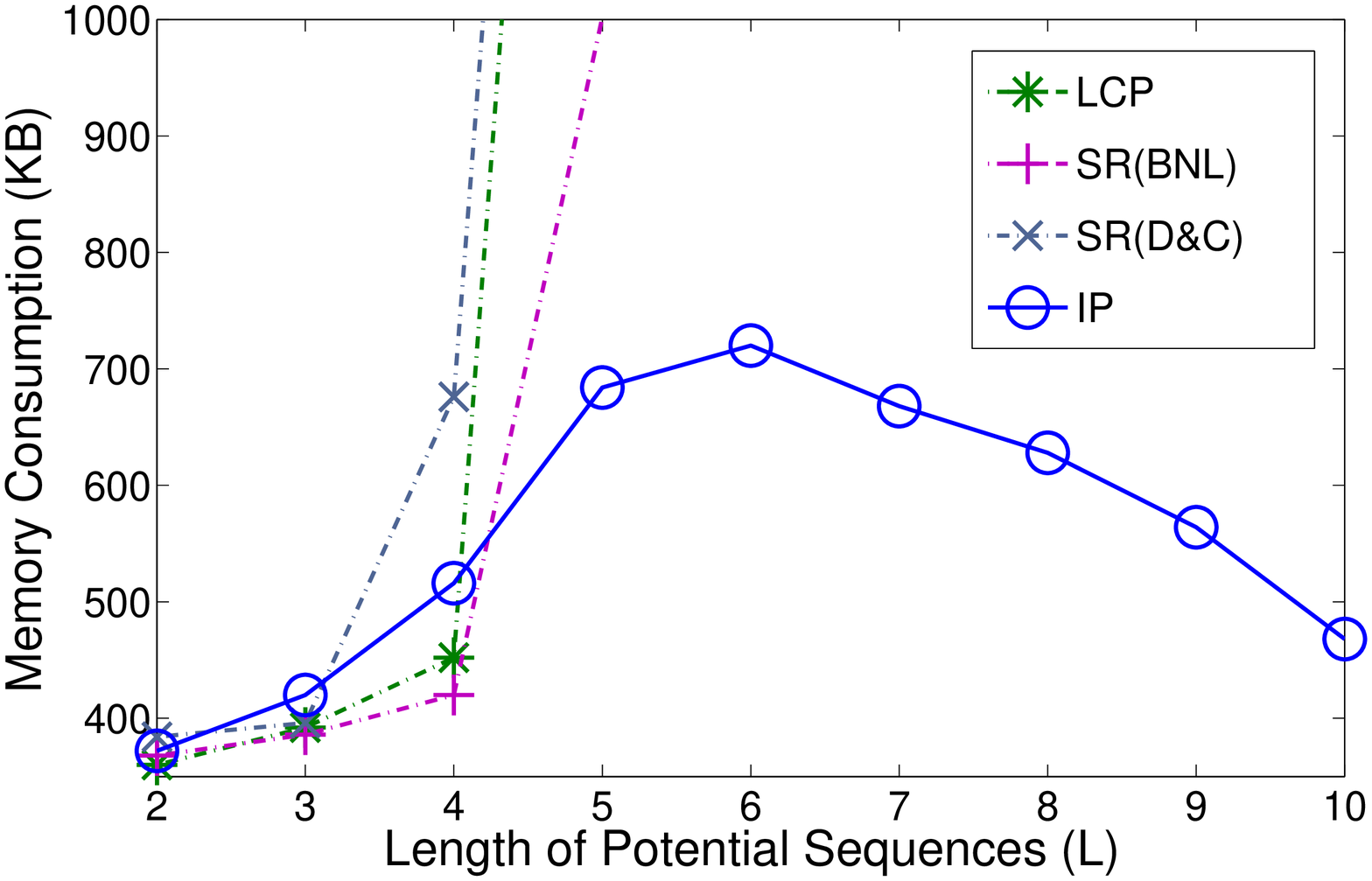}\label{MemC10}}
\subfigure[$|C|=20$]{\includegraphics[width=0.45\textwidth] {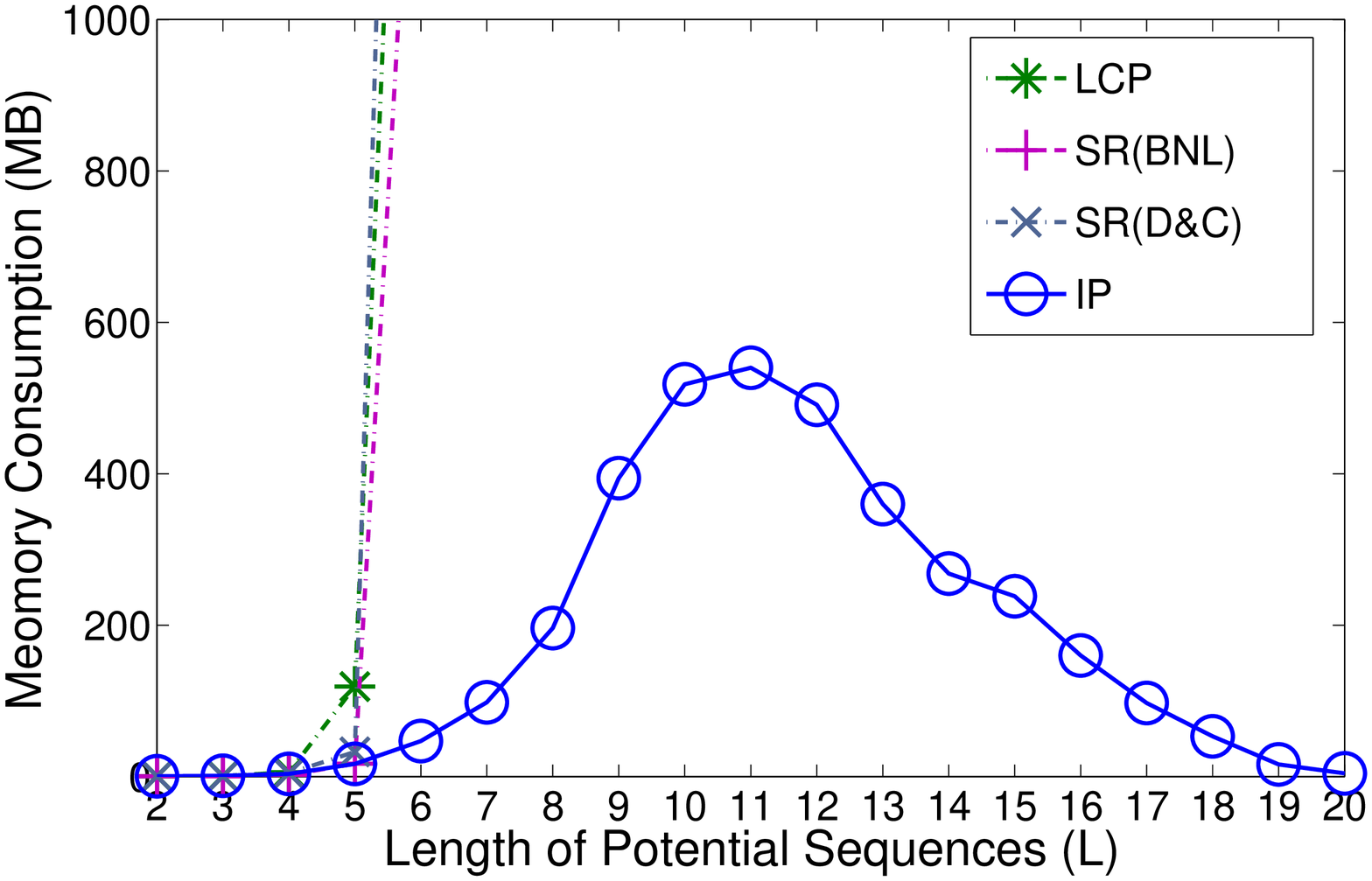}\label{MemC20}}
\caption{The internal memory consumption varying with the length of suggested driving routes on the synthetic data sets with $|C|=10$ and $|C|=20$. \label{MemC}}
\end{figure}
Figure~\ref{MemC} shows the varying of consumed internal memory storage with the length of potential sequence on synthetic data sets with $|C|=10$ and $|C|=20$ respectively. When the length $L \leq 4$, the memory cost of IP is a little higher than those of the other three algorithms due to the storage of some iterative calculation values, such as $F1$ and $PE$. However, as the length of potential sequence $L$ increases, the memory consumption of LCP, SR(BNL), and SR(D\&C) dramatically increase. Our algorithm IP presents an overall trend of parabola, which reaches the peak with 700K and 600M RAM on these two data sets respectively. Among algorithms LCP, SR(BNL) and SR(D\&C), the space cost of SR(D\&C) increases fastestly due to its recursive calculation process. In summary, the trends of the space cost of our algorithm IP on different data sets are consistent, and it is almost the same as the trend of the remaining sequence candidates.

Let us analyze the reason why the space cost of algorithm IP is relatively low. As we know, the generated sequence candidates and the associated values of $PE$ and $F1$ are stored in RAM only during the incremental process of generating the potential sequences from the length $L$ to $L+1$. Thus, the space cost is determined by the number of the remaining sequence candidates with the length $L$ and $L+1$ generated by IP. The number of the enumerated sequences dramatically increases with the increase of the number of pick-up points $N$ and the length of potential sequences $L$. Therefore, when the size of pick-up points $N$ becomes larger, the numbers of the remaining sequence candidates dramatically increase. Since the algorithm IP uses an incremental method to generate and reduce the potential sequences, its space performances are much better than those of BFS and LCP. Nevertheless, when the number of pick-up points is larger (e.g., $|C|=25$), the generation process of IP can not be performed in the internal memory of a PC with 4G RAM. In this case, we have to adopt external memory storage technology to generate the potential sequences.For algorithms BFS, LCP and SkyRoute, it is necessary to enumerate all possible sequences of a certain length, so the internal memory consumption is huge. As shown in Figure~\ref{MemC10}, they can only deal with the potential sequences with length \(L \le 5\) on the data set with $|C|=10$.
It can be observed that the memory consumption of our algorithm is really much lower than those of other methods.

\subsection{The Comparison of Online Search Time}
In this subsection, we compare the efficiency of various online route search algorithms. Note that all the results of search time come from the average values of 10 running cases. Tables~\ref{table3},~\ref{table4},~\ref{table5} and ~\ref{table6} show the online search time consumed by algorithms BFS, LCPS, SR(D\&C)S, SR(BNL)S, IPS and IBPS on both real-world and synthetic data sets with various numbers of pick-up points and lengths of suggested driving routes. It can be observed that the search time consumed by our algorithm IBPS is the least. The online route search of IPS is a little slower than that of IBPS. However, both of them always take a better performance over  the other four algorithms. For IPS and LCPS, as we know, when $L$ is small (e.g., $L=3$), the pruning ratio of IP is a little lower than that of LCP. However, our algorithm IP outperforms LCP benefited from the recursive computation of the PTD cost. The search time of the two skyline methods SR(D\&C)S and SR(BNL)S is much longer. The major reason is that skyline query is processed online which needs a rather long time. Therefore, this type of method is not suitable for recommending driving routes for a single cab. Actually, it performs better in providing multiple optimal driving routes for different cabs at the same place and time.

\begin{table*}[t]
\centering
\caption{A comparison of search time (millisecond) on the real-world data set (2-3PM).\label{table3}}
\begin{tabular}{|c|c|c|c|c|} \hline
&$L=2$&$L=3$&$L=4$&$L=5$ \\ \hline
BFS	&0.0077	&0.0286&	0.1980&	1.4341\\ \hline
LCPS&	0.0073&	0.0157&	0.0371&	0.1175\\ \hline
SR(D\&C)S&	10.3269	&41.6794&	182.0520&	1520.3100\\ \hline
SR(BNL)S&	1.9306&	24.3543&	139.0600&	2333.1200\\ \hline
IPS&	0.0070&	0.0110&	0.0165&	0.0246\\ \hline
IBPS&	0.0068&	0.0069	&0.0076&	0.0085\\ \hline
\end{tabular}
\end{table*}
\begin{table*}[t]
\centering
\caption{A comparison of search time (millisecond) on the synthetic data set with \(\left| C \right| = 15\) .\label{table4}}
\begin{tabular}{|c|c|c|c|c|} \hline
&$L=2$&$L=3$&$L=4$&$L=5$ \\ \hline
BFS &0.0125	&0.0925	&1.3028	&17.9584\\ \hline
LCPS&	0.0120&	0.0458&	0.3002&	1.9866\\ \hline
SR(D\&C)S&	24.7075	&154.6790&	1962.8100&	31210.4000\\ \hline
SR(BNL)S&	3.3707&	109.5560	&3612.4100&	210161.0000\\ \hline
IPS&	0.0089&	0.0556&	0.2317&	0.7322\\ \hline
IBPS&	0.0086&	0.0095&	0.0107&	0.0119\\ \hline
\end{tabular}
\end{table*}
\begin{table*}[t]
\centering
\caption{A comparison of search time (millisecond) on the synthetic data set with \(\left| C \right| = 20\).\label{table5}}
\begin{tabular}{|c|c|c|c|c|} \hline
&$L=2$&$L=3$&$L=4$&$L=5$ \\ \hline
BFS	&0.0189&	0.2244&	4.6117&	89.7804\\ \hline
LCPS&	0.0185&	0.1092&	0.8328&	7.9164\\ \hline
SR(D\&C)S&	40.1345	&617.7660&	11858.1000&	301341.0000\\ \hline
SR(BNL)S&	11.0259&	769.1580	&49492.5000&	5453600.0000\\ \hline
IPS	&0.0149	&0.0584&	0.3480&	1.5331\\ \hline
IBPS&	0.0120&	0.0141&	0.0172&	0.0202\\ \hline
\end{tabular}
\end{table*}
\begin{table*}[t]
\centering
\caption{A comparison of search time (millisecond) on the synthetic data set with \(\left| C \right| = 25\)\label{table6}}
\begin{tabular}{|c|c|c|c|c|} \hline
&$L=2$&$L=3$&$L=4$&$L=5$ \\ \hline
BFS	&0.0262&	0.4479&	11.9400	&305.5830\\ \hline
LCPS	&0.0260&	0.1919&	2.0199&	22.6318\\ \hline
SR(D\&C)S&	52.1815&	1362.7300&	36792.0000&	1093460.0000\\ \hline
SR(BNL)S&	17.7717	&1559.3200&	144124.0000	&24283100.0000\\ \hline
IPS&	0.0189&	0.1383&	0.8135&	5.1262\\ \hline
IBPS&	0.0148&	0.0182&	0.0242&	0.0446\\ \hline
\end{tabular}
\end{table*}
\begin{figure}[!htp]
\centering
\subfigure[]{\includegraphics[width=0.45\textwidth] {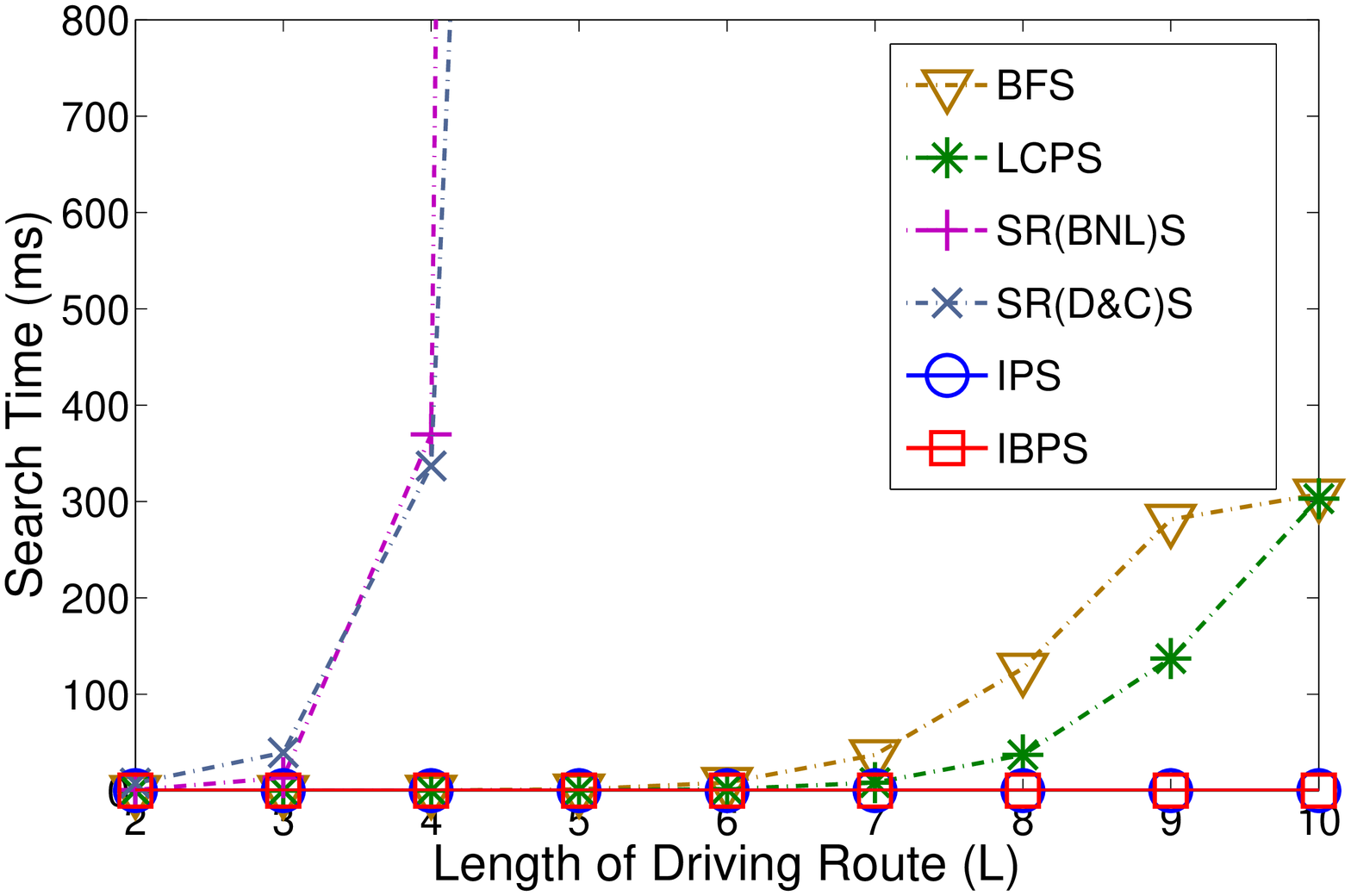}\label{route10searchtimeAll}}
\subfigure[]{\includegraphics[width=0.45\textwidth] {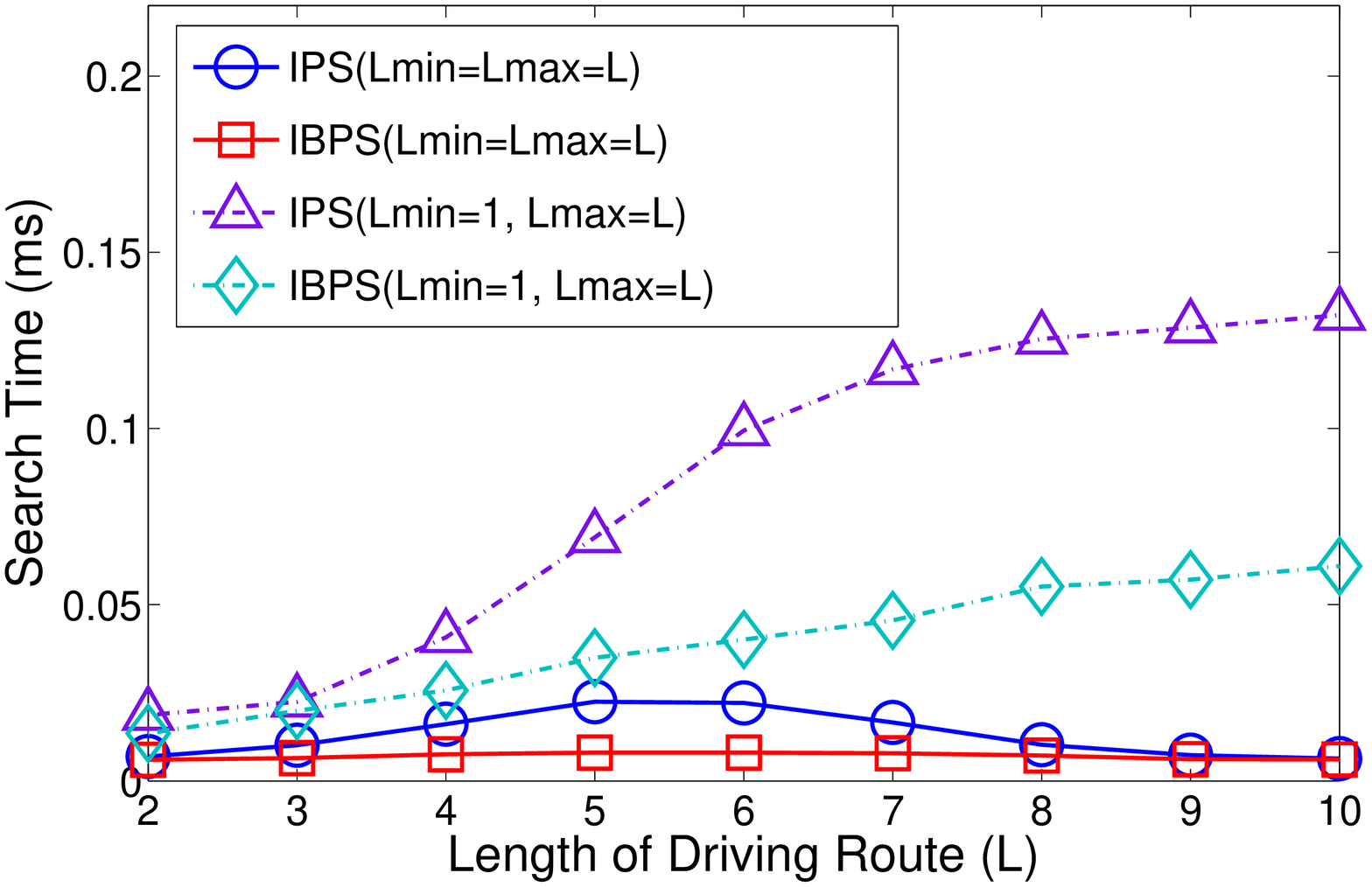}\label{searchtime10IPIBP}}
\caption{The search time varying with the length of suggested driving routes on the synthetic data set with $|C|=10$: (a) an overall comparison of various algorithms; (b) the comparison between IPS and IBPS.\label{searchtime}}
\end{figure}
Figure~\ref{searchtime} shows the curves of the search time varying with the length of suggested routes on the synthetic data set with $|C|=10$. A comparison of the search time for a certain length of driving routes $L$ of all the five algorithms above is given in Figure~\ref{route10searchtimeAll}. Obviously, the search time of SR(D\&C)S and SR(BNL)S dramatically increases along with the increase of the length of the suggested route. Since the number of remaining sequence candidates is very small, the search time consumed by our algorithms IBPS and IPS is always lower than those of other four algorithms, and it becomes more and more obvious as the length of suggested driving routes increases.

In addition, we add some significant tests by t-test when the length is small. Tables~\ref{table7} shows the average search time for LCPS, IPS and IBPS. The table also shows p values from a paired t-test for IPS and IBPS compared to LCPS. It can be observed that the search time consumed by our algorithms is always significantly lower than that of LCPS.

\begin{table*}[t]
\centering
\caption{The paired t-test compared to LCPS \label{table7}}
\begin{tabular}{|c|c|c|c|c|} \hline
&$L=2$&$L=3$&$L=4$&$L=5$ \\ \hline
mean(LCPS)	&0.0079&	0.0191&	0.0717	&0.3381\\ \hline
mean(IPS)	&0.0069&	0.0101&	0.0161	&0.0225\\ \hline
mean(IBPS)	&0.0060&	0.0065&	0.0076	&0.0081\\ \hline
t-test(IPS, LCPS)	&$p=0.001$ &$p=0.000$&$p =0.00$&$p=0.000$\\
&$p \ll 0.01$ &$p \ll 0.01$& $p \ll 0.01$&$p \ll 0.01$\\ \hline
t-test(IPS, LCPS)	&$p=0.000$ &$p=0.000$&$p=0.000$&$p=0.000$\\
&$p\ll 0.01$ &$p \ll 0.01$&$p \ll 0.01$&$p \ll 0.01$\\ \hline
\end{tabular}
\end{table*}

In order to make the trend of search time clearer, the curves of our algorithms with $L_{min}=L_{max}=L$ and $L_{min}=1, L_{max}=L$ on the synthetic data set with $|C|=10$ are shown in figure~\ref{searchtime10IPIBP} respectively. We can see that the search time of our algorithm IPS for a certain length of driving routes $L$ also shows a parabola trend, the same as the trend of the remaining sequence candidates. After the batch pruning, the number of remaining sequences becomes so small and the search time of IBPS is almost constant. The search time of our algorithms IPS and IBPS with $L_{min}=1$ and $L_{max}=L$ gradually increases with the increase of the maximal route length $L$. When $L=10$, the search time of our algorithms for all possible driving route with $1 \le L \le 10$ is still less than 0.14ms.

In summary, our online search algorithm has much lower search time compared to other existing methods. Moreover, it has a more flexible length constraint.

\section{Discussion}
In this section, we discuss some extensions of our method.

\subsection{Multiple Evaluation Functions}
As we know, the PTD function is a measure for evaluating the cost of a driving route. To meet different business requirements, we can adopt other evaluation functions. Two examples are given as follows.

\textbf{The Potential Travel Time (PTT)} \cite{29}. Since the driving time between two pick-up points usually depends on the traffic flow on the road, the distance does not always present the cost of a travel route properly. Thus, it is also valuable to recommend a route with least driving time. Let us give the definition of PTT. Assume that \(T_{c_{i-1} ,c_i }\) is the driving time from \(c_{i-1}\) to \(c_{i}\) during a certain period of time. In Formula~\ref{equ1}, if we replace the distance \(D_{c_{i-1} ,c_i }\) with travel time \(T_{c_{i-1} ,c_i }\) and \(D_{\infty}\) with \(T_{\infty}\), we can get a function of potential travel time

\begin{equation}
\label{equ11}
\scriptsize
F_T(\vec d) = T_{c_0 ,c_1 }  \cdot P(c_1 ) + (T_{c_0 ,c_1 }  + T_{c_1 ,c_2 } ) \cdot {\overline {P(c_1)}} \cdot P(c_2 ) +  \cdots  + T_\infty   \cdot \prod\limits_{i = 1}^L {\overline {P(c_i )} },
\end{equation}
where \(T_{\infty}\) denotes the desired maximum cruising time for a driver to pick up new passengers.

\textbf{The Potential Travel and Waiting Time (PTW).} In real life, the taxi drivers usually get passengers through two ways: cruising on a road and waiting in a place \cite{3,4}. Assume that a cab travels along a driving route \(\vec d = \left\langle {{c_0},{c_1},{c_2}, \ldots ,{c_L}} \right\rangle (1 \le L \le N) \), it has not gotten a passenger when arriving the last pick-up point \({c_L}\) and waits at \({c_L}\). Let the waiting time be a fixed value \({T_W}\) and the probability that it successfully gets a passenger at \({c_L}\) during the waiting time \({T_W}\) be \({P_W(c_L)}\), the cruising time vector of \(\vec d\) is
$
\small
T(\vec d) = \left\langle \begin{array}{l}
 T_{c_0 ,c_1 } ,(T_{c_0 ,c_1 }  + T_{c_1 ,c_2 } ), \ldots ,\sum\limits_{i=1}^L T_{c_{i-1} ,c_{i}} \\
 \end{array} \right\rangle$
and its probability vector is
$
\small
P(\vec d) = \left\langle \begin{array}{l}
 P(c_1 ),\overline {P(c_1 )}  \cdot P(c_2 ), \ldots ,\prod\limits_{i=1}^{L-1} \overline {P({c_i})}
 \cdot P(c_L )  \\
 \end{array} \right\rangle$. Then the time cost of successfully picking up a passenger by cruising is \(F_{C}(\vec d) = T(\vec d) \cdot P(\vec d)\).
The time cost of picking up a passenger by waiting at the last point \({c_L}\) is \(F_{W}(\vec d) = (\sum\limits_{i=1}^L T_{c_{i-1} ,c_{i}} + {T_W}) \cdot \prod\limits_{i=1}^{L} \overline {P({c_i})} \cdot {P_W(c_L)}\). The time cost when a driver does not get passengers after leaving the last point \({c_L}\) can be set to \({F_\infty}={T_\infty} \cdot \prod\limits_{i=1}^{L} \overline {P({c_i})} \cdot \overline {P_W(c_L)}\). Then, the PTW of route $\vec d$ can be given as

\begin{equation}
\label{equ12}
{F_{CW}(\vec d) =F_{C}(\vec d) + F_{W}(\vec d) + {F_\infty}}.
\end{equation}

\begin{figure}[!htp]
\centering
\subfigure[]{\includegraphics[width=0.36\textwidth] {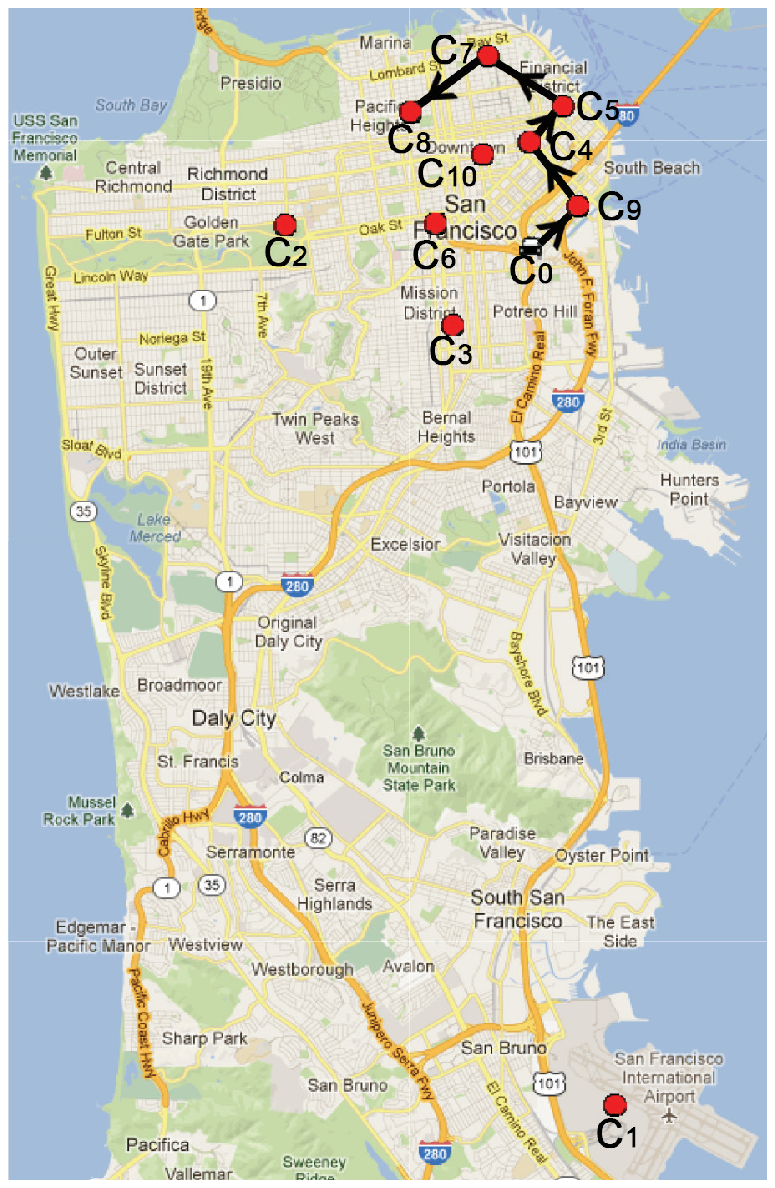}\label{clusters1}}
\subfigure[]{\includegraphics[width=0.36\textwidth] {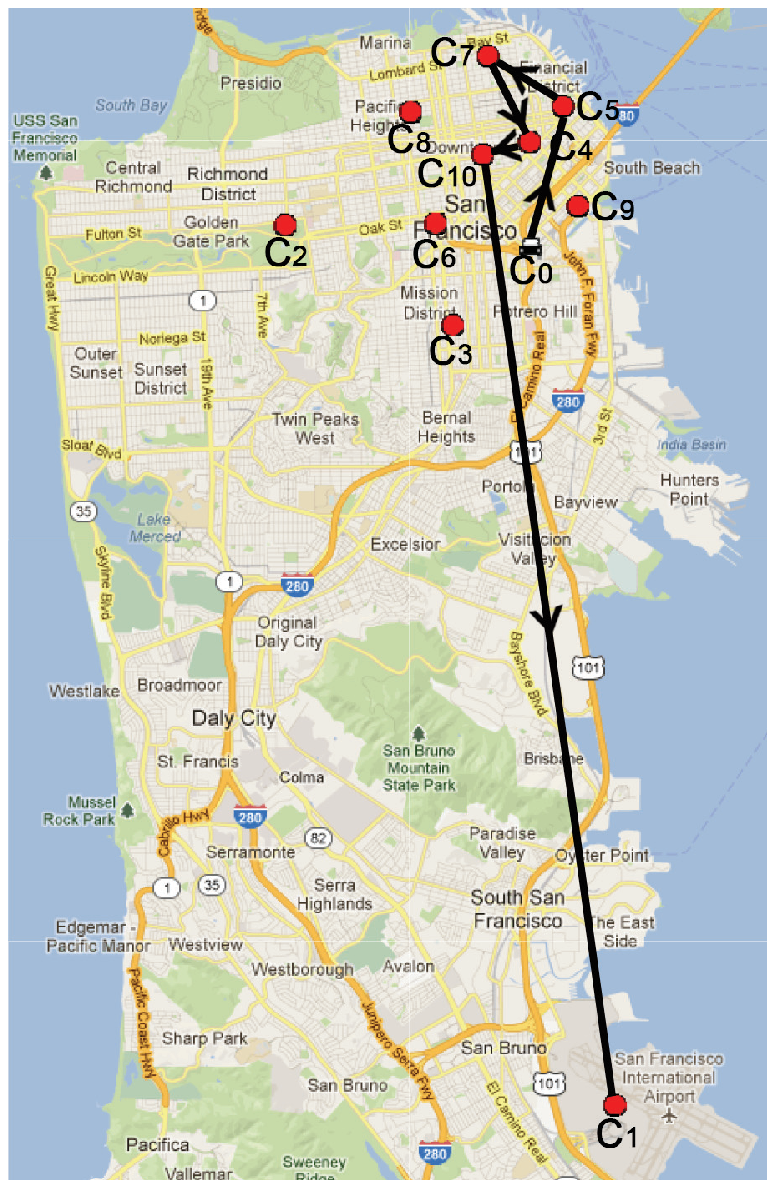}\label{clusters2}}
\vspace{-0.3cm}
\caption{The detected potential pick-up points and ${c_0}$ on real-world data set of 6-7PM.\label{clusters}}
\vspace{-0.2cm}
\end{figure}

With various evaluation functions, we can easily recommend different types of optimal driving routes to the drivers. For example, in Figure~\ref{clusters}, ten pick-up points $c_i (1 \le i \le 10)$ revealed on the real-world data set of 6PM-7PM and the current position $c_0$ of an empty cab are labeled on the map. When we set $L=3$, the optimal driving routes detected by PTD, PTT and PTW are $C0 \to C4 \to C5 \to C7$, $C0 \to C5 \to C7 \to C4$ and $C0 \to C5 \to C4 \to C7$, respectively. For $L=5$, the optimal driving routes evaluated by the PTD is  $C0 \to C9 \to C4 \to C5 \to C7 \to C8$, which is labeled in Figure~\ref{clusters1}  while the optimal driving routes evaluated by the PTT and PTW functions are the same: $C0 \to C5 \to C7 \to C4 \to C10 \to C8$. It is observed that the optimal drive routes are not always the same through different evaluation functions. Therefore, they can be applied to different applications.

\subsection{Recommendation with Destination Constraint}

Actually, our method can deal with the MSR problem with the destination point constraint. For example, if a driver wants to travel to a specified destination point, we can generate all possible potential sequences with the same source and destination points using our proposed algorithm BP-Growth. The slight difference is that we only perform the cost comparison and pruning among the potential sequences with the same source and destination points. Then we can recommend the optimal driving route satisfying the destination constraint to the driver online. Moreover, if the driver wants to wait for passengers at the destination point, we can consider the destination point as a temporal parking place and perform recommendation using the PTW measure presented above. For example, as shown in Figure~\ref{clusters2}, an optimal driving route revealed by the PTW function is $C0 \to C5 \to C7 \to C4 \to C10 \to C1$ when we set $C1$ as the destination, $L_{min}=3$ and $L_{max}=5$.

\subsection{Load Balance for Parallel Recommendations}

We briefly discuss how to make the recommendation be suitable for many cabs in the same area at the same time. For the generalized MSR problem, since both of the proposed algorithms IP and IBP deal with the potential sequences with the same source point, we can obtain the optimal driving routes starting from each pick-up point. Thus, we can get $N$ optimal driving routes with different source points. To perform the recommendation for multiple empty cabs simultaneously, we can adopt the load balancing techniques used in \cite{2}. The round-robin strategy maintains the number of the multiple empty cabs requesting the service, chooses one from the  N optimal driving routes  by the system in a circular manner \cite{26, 27} for the \(k_{th}\) request. For example, we can recommend the No.1 route with source point $c_1$ to the empty cab that first request the service, recommend the No.2 with route source point $c_2$ to the second empty cab, etc. And recommend the No.1 route again for the \((N+1)_{th}\) request.

\section{Conclusion}
This paper presents a dynamic programming based method to solve the problem of mobile sequential recommendation. The proposed method utilizes the iterative nature of the cost function and multiple pruning policies which greatly improve the pruning effect. The overall time complexity for handling mobile sequential recommendation problem without length constraint has been reduced from \(O(N!)\) to \(O({N^2} \cdot {2^N})\). Experimental results show that the pruning effect and the online search time are better than those of other existing methods. In the future, it will be interesting to use parallel algorithms for sequence generation and recommendation.

\section{Acknowledgements}

This research is supported by the Nature Science Foundation of China (No.61173093 and No.61202182). The authors would like to thank Prof. Jiawei Han and Dr. Jing Gao for their thoughtful
comments on this paper.

\bibliographystyle{plain}

\end{document}